\documentclass[a4paper,11pt]{article}

\usepackage{algorithm}
\usepackage{algpseudocode}
 \usepackage{microtype}
\usepackage{pdfpages}
\usepackage[T1]{fontenc}
 \usepackage{fullpage}
\usepackage{lmodern}
\usepackage{graphicx}
\usepackage{amsmath}
\usepackage{xspace}
\usepackage{amssymb}
\usepackage{amsthm}
\usepackage{tabularx}
\usepackage{multirow}
\usepackage{enumerate}
\usepackage{hyperref}

\newtheorem{theorem}{Theorem} 
\newtheorem{lemma}{Lemma}

\newtheorem{corollary}{Corollary}
\newtheorem{observation}{Observation}

\graphicspath{{./}}

\title{Farthest-point Voronoi diagrams in the presence of rectangular obstacles\thanks{
This research was partly supported by the Institute of Information \& communications Technology Planning \& Evaluation(IITP) grant funded by the Korea government(MSIT) (No. 2017-0-00905, Software Star Lab (Optimal Data Structure and Algorithmic Applications in Dynamic Geometric Environment)) and (No. 2019-0-01906, Artificial Intelligence Graduate School Program(POSTECH)).}} 

\author{Mincheol Kim\thanks{Department of Computer Science and Engineering, Pohang University of Science and Technology, Pohang, Korea. {\tt \{rucatia, sloth\}@postech.ac.kr}}
\and
Chanyang Seo\thanks{Graduate School of Artificial Intelligence, Pohang University of Science and Technology, Pohang, Korea. {\tt chan8616@postech.ac.kr}}
\and
Taehoon Ahn\footnotemark[2]
\and
Hee-Kap Ahn\thanks{Graduate School of Artificial Intelligence, Department of Computer Science and Engineering, Pohang University of Science and Technology, Pohang, Korea. {\tt heekap@postech.ac.kr}}
}

\DeclareMathOperator*{\argmax}{arg\,max}

\newcommand{\eps}{\ensuremath{\varepsilon}}

\newcommand{\stacktop}{\ensuremath{\mathsf{top}}}
\newcommand{\freespace}{\ensuremath{\mathsf{F}}}
\newcommand{\realspace}{\ensuremath{\mathbb{R}}}
\newcommand{\freespacenoempty}{\ensuremath{\mathsf{F}\setminus C_\emptyset}}

\newcommand{\sset}{\ensuremath{\mathsf{S}}}
\newcommand{\rectset}{\ensuremath{\mathsf{R}}}

\newcommand{\da}{\ensuremath{d_\alpha}}
\newcommand{\db}{\ensuremath{d_\beta}}
\newcommand{\dba}{\ensuremath{d_{\beta\alpha}}}

\newcommand{\tree}{\mathcal{T}}
\newcommand{\boundarylist}{\mathcal{B}}
\newcommand{\distlist}{\mathcal{D}}
\newcommand{\monolist}{\ensuremath{X}}

\newcommand{\nx}{\ensuremath{{x^-}}}
\newcommand{\px}{\ensuremath{{x^+}}}
\newcommand{\ny}{\ensuremath{{y^-}}}
\newcommand{\py}{\ensuremath{{y^+}}}
\newcommand{\rfm}{\ensuremath{\mathsf{M}}}
\newcommand{\rfmpy}{\ensuremath{\mathsf{M}_{y^+}}}
\newcommand{\rfmny}{\ensuremath{\mathsf{M}_{y^-}}}
\newcommand{\rfmpx}{\ensuremath{\mathsf{M}_{x^+}}}
\newcommand{\rfmnx}{\ensuremath{\mathsf{M}_{x^-}}}
\newcommand{\pl}{\ensuremath{{p}}}
\newcommand{\dspy}{\ensuremath{\mathsf{Q}_{y^+}}}
\newcommand{\dsny}{\ensuremath{\mathsf{Q}_{y^-}}}
\newcommand{\dspx}{\ensuremath{\mathsf{Q}_{x^+}}}
\newcommand{\dsnx}{\ensuremath{\mathsf{Q}_{x^-}}}

\newcommand{\boundbox}{\ensuremath{\mathsf{B}}}

\newcommand{\piru}{\ensuremath{\pi_{\mathsf{ru}}}}

\newcommand{\pilu}{\ensuremath{\pi_{\mathsf{lu}}}}
\newcommand{\pild}{\ensuremath{\pi_{\mathsf{ld}}}}
\newcommand{\pidl}{\ensuremath{\pi_{\mathsf{dl}}}}
\newcommand{\pidr}{\ensuremath{\pi_{\mathsf{dr}}}}
\newcommand{\pird}{\ensuremath{\pi_{\mathsf{rd}}}}

\newcommand{\fvd}{\ensuremath{\mathsf{FVD}}}
\newcommand{\mbisector}{\ensuremath{B}}
\newcommand{\trace}{\ensuremath{T}}
\newcommand{\face}{\ensuremath{\mathsf{f}}}

\begin{document}
\date{}
\thispagestyle{empty}
\maketitle

\begin{abstract}
We present an algorithm to compute the geodesic $L_1$ 
farthest-point Voronoi diagram of $m$ point sites 
in the presence of $n$ rectangular obstacles in the plane.
It takes $O(nm+n \log n + m\log m)$ construction time using $O(nm)$ space.
This is the first optimal algorithm for 
constructing the farthest-point Voronoi diagram in the presence of obstacles.
We can construct a data structure in the same construction time and space that answers a farthest-neighbor query in $O(\log(n+m))$ time.
\end{abstract}

\section{Introduction}
A Voronoi diagram of a set of sites is a subdivision of the space under consideration
into subspaces by assigning points to sites with respect to a certain proximity. 
Typical Voronoi assignment models are
the nearest-point model and the farthest-point model where every point is assigned to 
its nearest site and its farthest site, respectively.
There are results for computing Voronoi diagrams in the plane~\cite{aggarwal1989,edelsbrunner1986,fortune1987,shamos1975},
under different metrics~\cite{chew1985, klein1988, lee1980, papadopoulou2001},
or for various types of sites~\cite{alt2005,cheong2011,papadopoulou2013}.

For $m$ point sites in the plane, the nearest-point and farthest-point 
Voronoi diagrams of the sites can be constructed in $O(m\log m)$ time~\cite{fortune1987,shamos1975}.
When the sites are contained in a simple polygon with no holes,
the distance between any two points in the polygon, called 
the \emph{geodesic distance},
is measured as the length of the shortest path contained in the polygon 
and connecting the points (called the \emph{geodesic path}).
There has been a fair amount of work computing 
the geodesic nearest-point and farthest-point Voronoi diagrams of $m$ point sites 
in a simple $n$-gon~\cite{aronov1989,aronov1993,oa20geodesic,ola20fgeodesic}
to achieve the lower bound $\Omega(n+m\log m)$~\cite{aronov1989}.
Recently, optimal algorithms of $O(n+m\log m)$ time were given for 
the geodesic nearest-point Voronoi diagram~\cite{oh2019} 
and for the geodesic farthest-point Voronoi diagram~\cite{wang2021}. 

The problem of computing Voronoi diagrams is more challenging 
in the presence of obstacles. Each obstacle plays as a hole and 
there can be two or more geodesic paths connecting two points avoiding those holes.
The geodesic nearest-point Voronoi diagram of $m$ point sites 
can be computed in $O(m\log m+k\log k)$ time by applying the continuous Dijkstra paradigm~\cite{hershberger1999},
where $k$ is the number of total vertices of obstacles.
However, no optimal algorithm is known for the farthest-point Voronoi diagram 
in the presence of obstacles in the plane, even when the obstacles are of elementary
shapes such as axis-aligned line segments and rectangles.
The best result of the geodesic farthest-point Voronoi diagram
known so far takes $O(mk\log^2(m+k)\log k)$ time by Bae and Chwa~\cite{bae2009}.
They also showed that the total complexity of the geodesic farthest-point 
Voronoi diagram is $\Theta(mk)$.

In the presence of $n$ rectangular obstacles under $L_1$ metric,
there are some work for farthest-neighbor queries. 
Ben-Moshe et al.~\cite{moshe2001} presented a data structure 
with $O(nm\log(n+m))$ construction time and $O(nm)$ space for $m$ point sites 
that supports farthest point queries in $O(\log(n+m))$ time. 
They also showed that the $L_1$ geodesic farthest-point Voronoi diagram 
has complexity $\Theta(nm)$, but without presenting any algorithm for computing the diagram.
Later Ben-Moshe et al.~\cite{moshe2005} gave a tradeoff between the query
time and the preprocessing/space such that a data structure of size $O((n+m)^{1.5})$ can be constructed in $O((n+m)^{1.5} \log^2(n+m))$
to support farthest point queries in $O((n+m)^{0.5}\log(n+m))$ time. 

The geodesic center of a set of objects in a polygonal domain 
is the set of points in the domain that minimize the maximum geodesic distance from input objects. Thus, it can be obtained once the geodesic 
farthest-point Voronoi diagram of the objects is constructed. 
For $m$ points in the presence of $n$ axis-aligned rectangular obstacles in the plane,
Choi et al.~\cite{choi1998} showed that 
the geodesic center of the points under the $L_1$ metric consists of $\Theta(nm)$ 
connected regions and they gave an $O(n^2m)$-time algorithm to compute 
the geodesic center.
Later, Ben-Moshe et al.~\cite{moshe2001} gave an $O(nm \log (n+m))$-time algorithm for
the problem.

\subparagraph*{Our Result.}
In this paper, we present an algorithm that computes the geodesic $L_1$ farthest-point Voronoi diagram of $m$ points
in the presence of $n$ rectangular obstacles in the plane
in $O(nm+n\log n+m\log m)$ time using $O(nm)$ space.
The running time and space complexity of our algorithm match the time and 
space bounds of the Voronoi diagram. Thus, it is the first optimal algorithm 
for computing the geodesic farthest-point Voronoi diagram 
in the presence of obstacles.

To do this, we construct a data structure for $L_1$ farthest-neighbor queries in $O(nm+n \log n + m\log m)$ time using $O(nm)$ space.
This improves upon the results by Ben-Moshe et al.~\cite{moshe2001}, and
the construction time and space are the best among the data structures supporting 
$O(\log (n+m))$ query time for $L_1$ farthest neighbors.
Then we present an optimal algorithm to compute the explicit geodesic 
$L_1$ farthest-point Voronoi diagram in $O(nm+n\log n+m\log m)$ time using 
$O(nm)$ space, which matches the time and space lower bounds of the
diagram.

As a byproduct, we compute the geodesic center under the $L_1$ metric 
in $O(nm+n\log n+m\log m)$ time.
This result improves upon the algorithm by Ben-Moshe et al.~\cite{moshe2001}.

\subparagraph*{Outline.} 
First, we construct four farthest-point maps, one for each of the four axis directions, 
either the $x$- or $y$-axis, and either positive or negative.
In the course, we construct a data structure for $L_1$ farthest-neighbor queries in
$O(nm+n \log n + m\log m)$ time using $O(nm)$ space.
For each axis direction, 
we apply the plane sweep technique with a line orthogonal
to the direction and moving along the direction.
During the sweep, we maintain the status of the sweep line in a balanced 
binary search tree and its associated structures while handling events
induced by the point sites and the sides of rectangles parallel to the sweep line.
There are $m$ events induced by point sites and $O(n)$
events induced by rectangles.
After sorting the events in $O(n\log n + m\log m)$ time,
we show that we can handle all events induced by point sites in $O(nm)$ time.
Additionally, we show that each event induced by a rectangle can be handled in $O(m+\log n)$ time.
By the plane sweep, we construct a data structure
consisting of $O(n+m)$ line segments parallel to the sweep line and $O(nm)$ points
in $O(nm+n \log n + m\log m)$ time in total.
Given a query, it uses axis-aligned ray shooting queries on the data structure to
find the farthest site from the query.
The four farthest-point maps are planar subdivisions, and they can be constructed 
during the plane sweep in the same time and space.

With the four farthest-point maps and the data structure for farthest-neighbor queries,
we construct the geodesic $L_1$ farthest-point Voronoi diagram explicitly.
First, we decompose the plane, excluding the holes,
into rectangular faces using vertical line segments, each extended from 
a vertical side of a hole. 
Then, we partition each face in the decomposition
into zones such that the farthest-point Voronoi diagram restricted to
a zone coincides with the corresponding region of a farthest-point map.
This partition is done by using the boundary between two farthest-point maps,
which 
can be computed by traversing the cells in the two maps in which the boundary lies. 
Finally, we glue the corresponding regions along the boundaries of zones,
and then glue all adjacent faces along their boundaries to
obtain the geodesic $L_1$ farthest-point Voronoi diagram.
We show that this can be done in $O(nm+n\log n+m\log m)$ time in total.

For the centers of $m$ points in the presence of $n$ axis-aligned rectangles in the plane,
we can find them from the farthest-point Voronoi diagram 
in time linear to the complexity of the diagram.

\section{Preliminaries}
Let $\rectset$ be a set of $n$ open disjoint rectangles and 
$\sset$ be a set of $m$ point sites lying in the \emph{free space} $\freespace=\realspace^2 - \bigcup_{R\in\rectset}R$. 
We consider the $L_1$ metric. For ease of description, 
we omit $L_1$.
We use $x(p)$ and $y(p)$ to denote the $x$-coordinate and $y$-coordinate of a point $p$, respectively.
For two points $p$ and $q$ in $\freespace$, we use $pq$ to denote the line segment
connecting them.
Whenever we say a path connecting two points in $\freespace$, 
it is a path contained in $\freespace$.
There can be more than one geodesic path connecting two points
$p$ and $q$ avoiding the holes.
We use $\pi(p,q)$ to denote a fixed geodesic path connecting $p$ and $q$,
and use $d(p,q)$ to denote the geodesic distance between $p$ and $q$, which
is the length of $\pi(p,q)$.

We make a general position assumption that no point in $\freespace$ 
is equidistant from four or more distinct sites.
We use $f(p)$ to denote the set of sites of $\sset$ that are 
farthest from a point $p\in\freespace$ under the geodesic distance,
that is, a site $s$ is in $f(p)$ if and only if $d(s,p)\ge d(s',p)$ for all $s'\in\sset$.
If there is only one farthest site, we use $f(p)$ to denote the site.

A horizontal line segment $\ell$ can be represented by
the two $x$-coordinates $x_1(\ell)$ and $x_2(\ell)$ of its endpoints ($x_1(\ell) < x_2(\ell)$)
and the $y$-coordinate $y(\ell)$ of them.
For an axis-aligned rectangle $R$,
let $x_1(R)$ and $x_2(R)$ denote the $x$-coordinates of the left and
right sides of $R$. 

A path is \emph{$x$-monotone} if and only if the intersection of the path with any line
perpendicular to the $x$-axis is connected.
Likewise, a path is \emph{$y$-monotone} if and only if the intersection of the path
with any line perpendicular to the $y$-axis is connected.
A path is \emph{$xy$-monotone} if and only if the path is $x$-monotone and $y$-monotone.
Observe that if a path connecting two points is $xy$-monotone, it is 
a geodesic path connecting the points.

\begin{figure}[t]
	\begin{center}
	  \includegraphics[width=\textwidth]{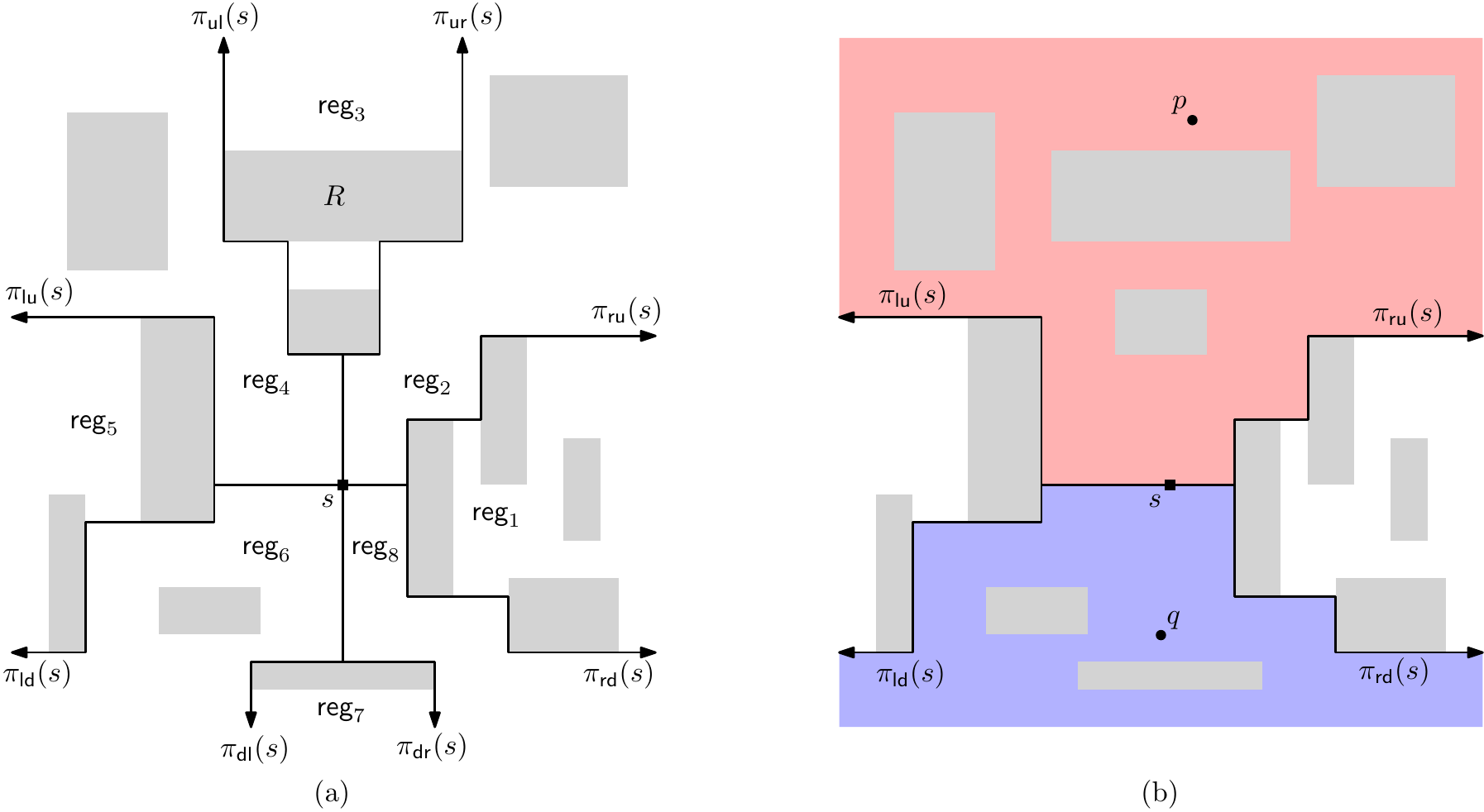}
	  \caption{\small
Gray rectangles are holes.
(a) The eight paths partition $\freespace$ into eight regions $\textsf{reg}_1,\ldots,\textsf{reg}_8$. Region $\textsf{reg}_3$ consists of two regions separated by a rectangle $R$.
(b) Every geodesic path from $s$ to $p$ is $\py$-monotone and $p$ is $\py$-reachable from $s$. Every geodesic path from $s$ to $q$ is $\ny$-monotone and $q$ is $\ny$-reachable from $s$.
  }
	  \label{fig:monotone}
	\end{center}
   \end{figure}

\subsection{Eight Monotone Paths from a Point}
Choi and Yap~\cite{choi1996} gave a way of partitioning 
the plane with rectangular holes
into eight regions using eight $xy$-monotone paths from a point. 
We use their method to partition $\freespace$ as follows.
Consider a horizontal ray emanating from a point $s=p_1\in\freespace$ going rightwards.
The ray stops when it hits a rectangle $R\in\rectset$ at a point $p_1'$. 
Let $p_2$ be the top-left corner of
$R$. We repeat this process by taking a horizontal ray from
$p_2$ going rightwards until it hits a rectangle, and so on. 
Then we obtain an $xy$-monotone path
$\piru(s)=p_1p_1'p_2p_2'\ldots$
from $s$ that alternates going \emph{rightwards} and going \emph{upwards}. 

By choosing two directions, one going either rightwards or leftwards horizontally,
and one going either upwards or downwards vertically, and ordering the chosen directions, 
we define eight rectilinear $xy$-monotone paths with directions:
rightwards-upwards (\textsf{ru}), upwards-rightwards (\textsf{ur}), upwards-leftwards (\textsf{ul}), leftwards-upwards (\textsf{lu}), 
leftwards-downwards (\textsf{ld}), downwards-leftwards (\textsf{dl}), downwards-rightwards (\textsf{dr}),
and rightwards-downwards (\textsf{rd}).
Let $\pi_\delta(s)$ denote one of the eight paths corresponding to
the direction $\delta$ in
$\{\mathsf{ru}, \mathsf{ur}, \mathsf{ul}, \mathsf{lu}, \mathsf{ld}, \mathsf{dl}, \mathsf{dr}, \mathsf{rd}\}$.

Some of the eight paths $\pi_\delta(s)$ may overlap in the beginning 
from $s$ but they do not cross each other. The paths partition 
$\freespace$ into eight regions $\textsf{reg}_1,\ldots,\textsf{reg}_8$ with 
the indices sorted around $s$ in a counterclockwise order
such that $\textsf{reg}_1$ denotes the region lying to the right of $s$, 
below $\piru(s)$ and above $\pird(s)$.
Observe that $\textsf{reg}_i$ is not necessarily connected.
See Figure~\ref{fig:monotone}(a) for an illustration. 

\begin{lemma}[\cite{choi1996,rezende1985}]\label{lem:monotone}
Every geodesic path connecting two points is either $x$-, $y$-, or $xy$-monotone.
For a point $s\in\freespace$, following three statements hold.
\begin{itemize}
\item If $p\in\textsf{reg}_1\cup\textsf{reg}_5$,
every geodesic path from $s$ to $p$ is $x$-monotone but not $y$-monotone.
\item If $p\in\textsf{reg}_3\cup\textsf{reg}_7$, 
every geodesic path from $s$ to $p$ is $y$-monotone but not $x$-monotone.
\item If $p\in\textsf{reg}_2\cup\textsf{reg}_4\cup\textsf{reg}_6\cup\textsf{reg}_8\cup\Pi(s)$, 
every geodesic path from $s$ to $p$ is $xy$-monotone, 
where $\Pi(s)$ is the union of the eight paths $\pi_\delta(s)$.
\end{itemize}
\end{lemma}

Based on Lemma~\ref{lem:monotone}, we define a few more terms.
For any point $p$ in $\textsf{reg}_2\cup\textsf{reg}_3\cup\textsf{reg}_4$ 
(and the boundaries of the regions), we say $p$ is \emph{$\py$-reachable} from $s$, and
every geodesic path from $s$ to $p$ is \emph{$\py$-monotone.}
Any point $q\in\textsf{reg}_6\cup\textsf{reg}_7\cup\textsf{reg}_8$ 
(and the boundaries of the regions) is \emph{$\ny$-reachable} from $s$, 
and every geodesic path  from $s$ to $q$ is \emph{$\ny$-monotone}. 
See Figure~\ref{fig:monotone}(b).
Similarly, any point $p\in\textsf{reg}_1\cup\textsf{reg}_2\cup\textsf{reg}_8$ 
(and the boundaries of the regions) is \emph{$\px$-reachable} from $s$, and
every geodesic path  from $s$ to $p$ is \emph{$\px$-monotone}. 
Any point $q\in\textsf{reg}_4\cup\textsf{reg}_5\cup\textsf{reg}_6$ 
(and the boundaries of the regions) is \emph{$\nx$-reachable} from $s$, and
every geodesic path  from $s$ to $q$ is \emph{$\nx$-monotone}.

\section{Farthest-point Maps}\label{sec:fpm}
Based on Lemma~\ref{lem:monotone} and the four directions of monotone paths 
in the previous section, we define four \emph{farthest-point maps}.
A farthest-point map $\rfmpy = \rfmpy(\sset)$ of $\sset$ in $\freespace$ corresponding to 
the positive $y$-direction is a planar subdivision of $\freespace$ into cells.
For a point $p\in\freespace$, a site $s\in\sset$ is a farthest site of $p$ 
in $\rfmpy$ if $d(p,s)\ge d(p,s')$ for every site $s'\in\sset$ 
from which $p$ is $\py$-reachable.
If $p$ is $\py$-reachable from no site in $\sset$, 
$p$ has no farthest site in $\rfmpy$.
Thus, a cell of $\rfmpy$ is defined on $\freespacenoempty$, 
where $C_\emptyset$ denotes the set of points of $\freespace$ 
that are $\py$-reachable from no site in $\sset$.
A site $s$ corresponds to one or more cells in $\rfmpy$ with the property 
that a point $p\in\freespacenoempty$ lies in a cell of $s$ 
if and only if $d(p,s) > d(p,s')$
for every $s'\in\sset\setminus\{s\}$ from which $p$ is $\py$-reachable.


We define $\rfmny$, $\rfmpx$ and $\rfmnx$ analogously 
with respect to their corresponding directions.
Since the four maps have the same structural and combinatorial properties 
with respect to their corresponding directions,
we describe only $\rfmpy$ in the following.
Let $\boundbox$ be an axis-aligned rectangular box such that
$\sset$, $\rectset$, and all vertices of the four farthest-point maps are contained 
in the interior of $\boundbox$.
We focus on $\freespace\cap \boundbox$ only, and use $\freespace$ as $\freespace \cap \boundbox$.

In the following, we analyze the edges of $\rfmpy$ using the bisectors of pairs of sites.
Let $F(s, s')$ denote a set of points of $\freespace$ that are $\py$-reachable from two sites $s$ and $s'$. 
To be specific, $F(s,s')$ is an intersection of two regions, 
one lying above $\pilu(s)$ and $\piru(s)$ and the other lying above $\pilu(s')$ and $\piru(s')$.
Thus, the boundary of $F(s, s')$ coincides with 
the upper envelope of $\pilu(s)$, $\piru(s)$, $\pilu(s')$ and $\piru(s')$.
We use $F(s,s)$ to denote the set of points that are $\py$-reachable 
from a site $s$.

For any two distinct sites $s, s'\in\sset$,
their \emph{bisector} consists of all points 
$x\in\freespace$ satisfying $\{x\mid d(x,s)=d(x,s')\}$.
Observe that the bisector may contain a two-dimensional region.
We use $b(s, s')$ to denote the line segments and the boundary 
of the two-dimensional region in the bisector of $s$ and $s'$. 

\begin{lemma}\label{lem:boundary}
For any two sites $s$ and $s'$, $b(s,s') \cap F(s,s')$ consists of axis-aligned segments.
\end{lemma}
\begin{proof}
Consider two sites $s$ and $s'$ in the plane with no holes. 
Then $b(s,s')$ contained in $\boundbox$ is 
a polygonal chain consisting of two parallel and axis-aligned segments, 
and one segment of slope $1$ or $-1$ lying in between them.
The segment of slope $-1$ appears in region $[x(s),x(s')]\times[y(s),y(s')]$
if $x(s)\le x(s')$ and $y(s)\le y(s')$, and the segment of slope $1$ appears 
in region $[x(s),x(s')]\times[y(s),y(s')]$ if $x(s')\le x(s)$ and $y(s)\le y(s')$.

Now consider the bisector $b(s,s')$ of two sites $s$ and $s'$ in the freespace $\freespace$. 
Due to the rectangle holes of $\freespace$, 
the bisector may consist of two or more pieces.
It, however, still consists of axis-aligned segments,
and segments of slope $\pm 1$ under the $L_1$ metric~\cite{mitchell1992}.

We now focus on $b(s,s')$ restricted to $F(s,s')$ and
show that no segment of slope $\pm 1$ appears in $b(s,s') \cap F(s,s')$.
Assume to the contrary that $b(s,s') \cap F(s,s')$ has a segment 
$uv$ of slope $1$ or $-1$. 
Let $q$ be any point on $uv$, and
$q_1$ be the farthest point from $q$ on $\pi(s,q)$ such that $\pi(q_1,q)$ is $xy$-monotone.
Likewise,
let $q_2$ be the farthest point from $q$ on $\pi(s',q)$ such that $\pi(q_2,q)$ is 
$xy$-monotone.
Clearly, $d(s,q_1)+d(q_1,q)=d(s',q_2)+d(q_2,q)$.
Since $uv$ has slope $1$ or $-1$,
$\min\{x(q_1),x(q_2)\}\leq x(q) \leq \max\{x(q_1),x(q_2)\}$ and
$\min\{y(q_1),y(q_2)\}\leq y(q) \leq \max\{y(q_1),y(q_2)\}$.
Then $q$ is $\py$-reachable from one of $q_1$ and $q_2$,
and $\ny$-reachable from the other,
implying that $\pi(s,q)$ or $\pi(s',q)$ is not $y$-monotone, a contradiction.
Thus, no segment of slope $1$ appears in $b(s,s') \cap F(s,s')$.
Analogously, we can show that $b(s,s') \cap F(s,s')$ 
has no segment with slope $-1$.
\end{proof}

\begin{figure}[t]
	\begin{center}
	  \includegraphics[width=\textwidth]{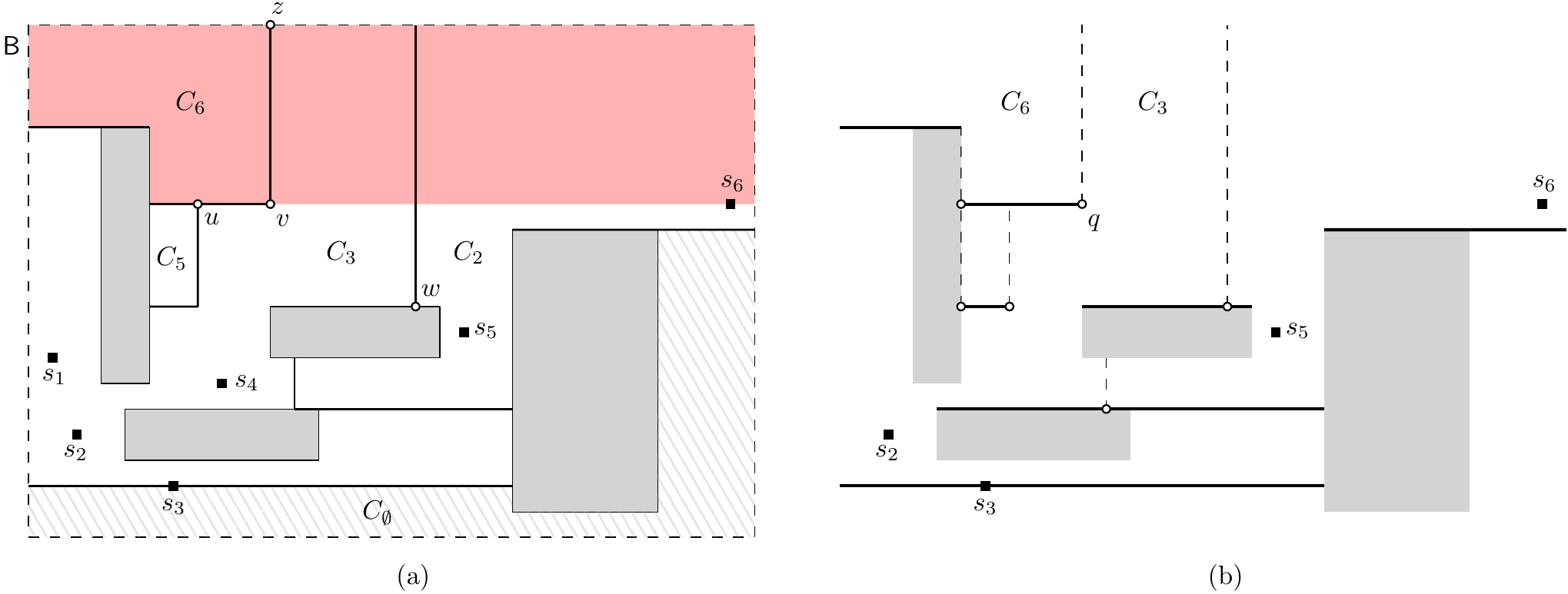}
	  \caption{\small
(a) $\rfmpy$ for $\sset=\{s_1,\ldots,s_6\}$
restricted to a box $\boundbox$ with four rectangular holes (gray).
$s_i$ has a corresponding cell $C_i$ for $i=2,3,5,6$ while $s_1$ and $s_4$ 
have no cell.
A vertical edge $vz$ is from $b(s_3,s_6)$ in the (red) region $F(s_3,s_6)$.
A horizontal edge $uv$ is not part of $b(s_3,s_6)$ but it is part of a $b$-edge 
as no point lying below $uv$ is $\py$-reachable from $s_6$.
(b) Illustration of $\dspy$ corresponding to $\rfmpy$.
At the boundary point $q$, $d(q,s_3) = d(q,s_6)$.
	  }
	  \label{fig:rfmpy}
	\end{center}
   \end{figure}

Let $f_\delta(p)$ denote the set of farthest sites from
a point $p\in\freespace$ among the sites from which $p$ is $\delta$-reachable for $\delta \in \{\py, \ny, \px, \nx\}$.
For each horizontal segment of $\pilu(s)\cup\piru(s)$,
we call the portion $h$ of the segment such that $f_\py(p)=\{s\}$ for any point $p\in h$,
a \emph{$b$-edge}.
Observe that no point $p'$ with $x_1(h)\le x(p') \le x_2(h)$ and $y(p') = y(h) - \eps$ for any $\eps > 0$ 
is $\py$-reachable from $s$.
Thus, a $b$-edge is also an edge of $\rfmpy$.
Since every edge of $\rfmpy$ is part of a bisector of two sites in $\sset$
or a $b$-edge, it is either horizontal or vertical.
See Figure~\ref{fig:rfmpy}(a).

\begin{corollary}\label{co:boundary}
Every edge of $\rfmpy$ is an axis-aligned line segment.
\end{corollary}


For sites contained in a simple polygon, Aronov~et~al.~\cite{aronov1993} gave 
a lemma, called \emph{Ordering Lemma},
that the order of sites along their convex hull is the same as the order of their Voronoi 
cells along the boundary of a simple polygon. 
We give a lemma on the order of sites in the presence of rectangular obstacles. 
We use it in analyzing the maps and Voronoi diagrams.

\begin{lemma}\label{lem:reverseorder}
Let $pq$ be a horizontal segment contained in $\freespacenoempty$ with $x(p) < x(q)$. 
For any two sites $f_p \in f(p)$ and $f_q \in f(q)$
such that $p$ and $q$ are $\py$-reachable from both $f_p$ and $f_q$, 
if $f_p \notin f(q)$ or $f_q \notin f(p)$, $x(f_p) > x(f_q)$.
\end{lemma}
\begin{proof}
Ben-Moshe et al.~\cite{moshe2001} showed that 
$x(f_p) \neq x(f_q)$.

Assume to contrary that $x(f_p) < x(f_q)$. Consider two cases that
there are (1) two geodesic paths, one from $p$ to $f_q$ and one from $q$ to $f_p$, 
intersecting each other at a point, say $r$ (Figure~\ref{fig:lemreverse}(a)),
or (2) no two such geodesic paths intersecting each other (Figure~\ref{fig:lemreverse}(b)).

For case (1), we have
\begin{equation*}
d(p,f_q) = d(p,r)+d(r,f_q)\quad \text{and}\quad d(q,f_p) = d(q,r)+d(r,f_p).
\end{equation*}
We also observe that 
\begin{equation*}
d(p,f_p) \leq d(p,r)+d(r,f_p)\quad \text{and} \quad d(q,f_q) \leq d(q,r)+d(r,f_q).
\end{equation*}
Adding the two inequalities above, we obtain 
\begin{eqnarray*}
d(p,f_p) + d(q,f_q) &\leq& d(p,r)+d(r,f_p)+d(q,r)+d(r,f_q)\\
 &=& d(p,f_q)+d(q,f_p).
\end{eqnarray*}
However, since $f_p \notin f(q)$ or $f_q \notin f(p)$,
we have $d(p,f_p) + d(q,f_q) > d(p,f_q)+d(q,f_p)$, a contradiction.

Now consider case (2) that there are no two geodesic paths, 
one from $p$ to $f_q$ and one from $q$ to $f_p$,
intersecting each other (Figure~\ref{fig:lemreverse}(b)).
Since $pq$ is horizontal with $x(p)<x(q)$, $x(f_p) < x(f_q)$ by assumption, and 
every geodesic path from $p$ to $f_q$ and every geodesic path from 
$q$ to $f_p$ are $\py$-monotone,
we have $y(f_p) \neq y(f_q)$.
Without loss of generality, assume $y(f_p) < y(f_q)$.
Then $\piru(f_q)$ intersects $\pi(q,f_p)$ at a point, say $r$ (Figure~\ref{fig:lemreverse}(c)). 
Since $x(f_p) < x(f_q) \leq x(r)$ and $y(f_p) < y(f_q) \leq y(r)$, 
we have $d(r,f_q)<d(r,f_p)$, and thus 
\begin{equation*}
d(q,r) + d(r,f_q) < d(q,r) + d(r,f_p) = d(q,f_p).
\end{equation*}
Since $d(q,f_q) \leq d(q,r) + d(r,f_q)$, we have $d(q,f_q)<d(q,f_p)$, a contradiction.
\begin{figure}[t]
  \begin{center}
    \includegraphics[width=\textwidth]{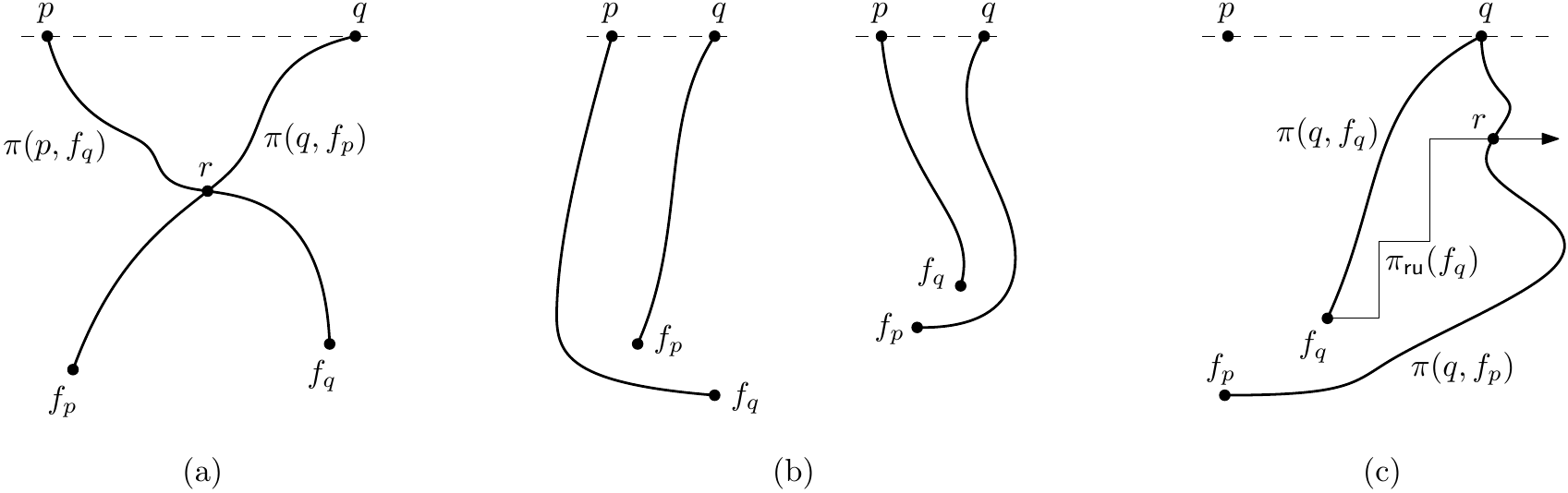}
    \caption{\small
Two cases (a) and (b) for two geodesic path $\pi(p,f_q)$ and $\pi(q,f_p)$ when $x(f_p) < x(f_q)$.
(a) $\pi(p,f_q)$ and $\pi(q,f_p)$ intersect each other at $r$.
(b) $\pi(p,f_q)$ and $\pi(q,f_p)$ do not intersect each other.
(c) If $x(f_p)<x(f_q)$ and $y(f_p)\le y(f_q)$, $\piru(f_q)$ intersects $\pi(q,f_p)$ at $r$.
    }
    \label{fig:lemreverse}
  \end{center}
\end{figure}
\end{proof}

Since there are at most $m$ sites, we obtain the following corollary from Lemma~\ref{lem:reverseorder}.

\begin{corollary}\label{co:reverse}
Any horizontal line segment contained in $\freespace$ intersects at most $m$ cells in $\rfmpy$.
\end{corollary}


Using Corollary~\ref{co:boundary} and~\ref{co:reverse},
we analyze the complexity of $\rfmpy$ as follows.
Note that each lower endpoint of a vertical edge of $\rfmpy$ appears
on a horizontal line segment passing through a site or the top side of a rectangle.
By Corollary~\ref{co:reverse}, the maximal horizontal segment through
the top side of a rectangle in $\rectset$ and contained in $\freespace$
intersects $O(m)$ vertical edges of $\rfmpy$.
Moreover, 
the maximal horizontal line segment through a site $s$
and contained in $\freespace$ 
intersects $O(1)$ lower endpoints of vertical edges on the boundary 
of the cell of $s$.
Since there are $n$ rectangles in $\rectset$ and $m$ sites in $\sset$, 
$\rfmpy$ has $O(nm+m) = O(nm)$ vertical edges.
Every horizontal edge of $\rfmpy$ is a segment of a bisector or a $b$-edge,
and it is incident to a side of a rectangle or another vertical edge.
Since there are $O(n)$ rectangle sides, and 
$O(1)$ horizontal edges of $\rfmpy$ that are incident to a vertical edge,
$\rfmpy$ has $O(n + nm) = O(nm)$ horizontal edges.
Thus, $\rfmpy$ has complexity $O(nm)$.

Now we show that every farthest site $s\in f(p)$ of a point $p$ in $\freespace$ 
is one of the farthest sites of $p$ in the four farthest-point maps.
By the definition of the farthest-point maps, $p$ 
is contained in a cell of 
$\rfmpy$, $\rfmny$, $\rfmpx$ or $\rfmnx$.
Since every geodesic path connecting two points is either $\py$-, $\ny$-, $\px$-, or $\nx$-monotone
by Lemma~\ref{lem:monotone}, 
$s\in f(p)$ is one of the farthest sites of $p$ in the four farthest-point maps.
If $p$ is contained in cells of two or more 
maps,
we compare their distances to the farthest sites defining the cells and
take the ones with the largest distance as the farthest sites of $p$.
Thus, once the four farthest-point maps are constructed,
the farthest sites of a query point can be computed from
the map.

\section{Data Structure for Farthest-neighbor Queries}\label{sec:algorithms}
We present an algorithm that constructs a data structure for farthest
site queries. 
We denote $m$ point sites of $\sset$ by $s_1,\ldots,s_m$ 
such that $x(s_1) \leq \cdots \leq x(s_m)$,
and $n$ rectangular obstacles of $\rectset$ by $R_1,\ldots,R_n$.
The data structure consists of four parts, each for one axis direction.
Since the four parts can be constructed in the same way with respect to their directions,
we focus on the part corresponding to the positive $y$-direction,
and thus the structure corresponds to $\rfmpy$. We use $\dspy$ to denote
the query data structure. 

By Corollary~\ref{co:boundary}, we can find the farthest site of a query point 
using a vertical ray shooting query to the horizontal edges 
of $\rfmpy$ and a binary search on the lower endpoints of vertical edges of $\rfmpy$
lying on the horizontal edges of $\rfmpy$.
Thus, we construct $\dspy$ such that it consists of the horizontal edges
of $\rfmpy$ and the endpoints of vertical edges of $\rfmpy$ 
lying on the horizontal edges of $\rfmpy$.

A point $q$ lying on a horizontal segment $h$ of $\dspy$ 
is the lower endpoint of a vertical edge of $\rfmpy$ if and only if
there are two points $q_1=(x(q) - \eps,y(q))$ and $q_2=(x(q) + \eps,y(q))$
for sufficiently small $\eps > 0$ satisfying 
$f_\py(q_1) \cup f_\py(q_2) = f_\py(q)$ and $f_\py(q_1) \neq f_\py(q_2)$.
We call each lower endpoint of vertical edges lying on $h$ 
a \emph{boundary point} on $h$. 
See Figure~\ref{fig:rfmpy}(b).

\begin{figure}[t]
  \begin{center}
    \includegraphics[width=.85\textwidth]{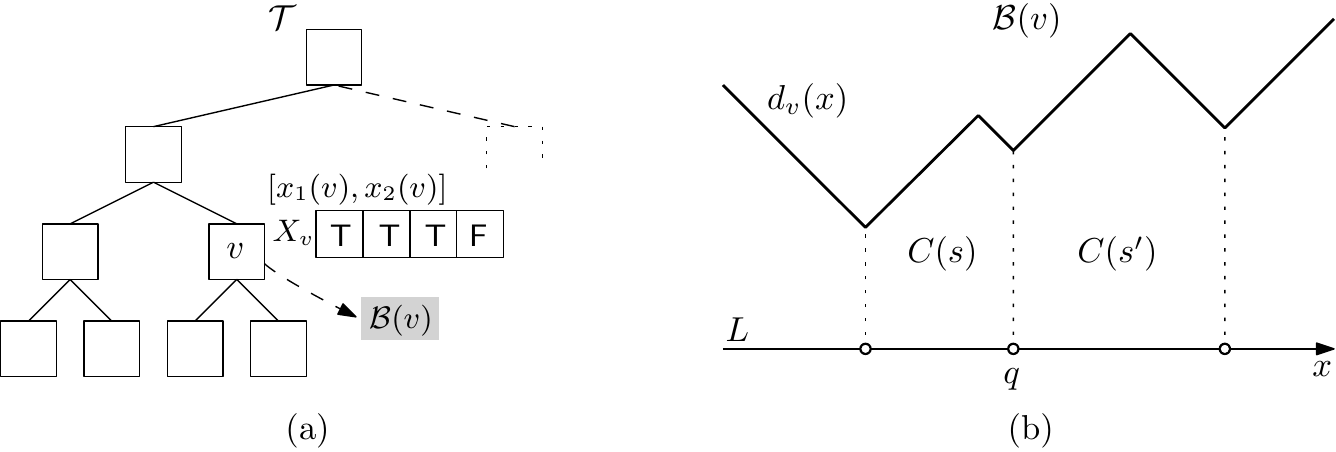}
    \caption{\small
(a) Illustration of a balanced binary search tree $\tree$. A node $v$ in $\tree$ has 
domain $[x_1(v), x_2(v)]$, array $\monolist_v$, and a pointer to $\boundarylist(v)$.
(b) Illustration of $\boundarylist(v)$ and $d_v(x)$.
}
    \label{fig:planesweep}
  \end{center}
\end{figure}


We use a plane sweep algorithm with a horizontal sweep line $L$
to construct the horizontal line segments in $\dspy$.
Note that $\freespace \cap L$ consists of disjoint horizontal segments
along $L$. The status of $L$ is the sequence of segments in 
$\freespace\cap L$ along $L$. 
The status changes while $L$ moves upwards over the plane, but not continuously.
Each update of the status occurs at a particular $y$-coordinate, 
which we call an \emph{event}.
To do such updates efficiently, we maintain three data structures for $L$: 
a \emph{balanced binary search tree} $\tree$ representing the status,
a \emph{boundary list} $\boundarylist$, and a list $\distlist$ of distance functions.
The structures $\boundarylist$ and $\distlist$ are associated structures of $\tree$.

We store the segments of $\freespace \cap L$
in a balanced binary search tree $\tree$ in increasing order of
$x$-coordinate of their left endpoints.
Each node $v$ of $\tree$ corresponds to a horizontal line segment 
$h_v$ of $\freespace \cap L$. 
We store $x_1(h_v)$ and $x_2(h_v)$,
and an array $\monolist_v$ of $m$ Boolean variables at $v$.
We set $\monolist_v[i] = \texttt{T}$ if a point on $h_v$ is 
$\py$-reachable from $s_i$ for $i=1\ldots,m$. 
Otherwise, we set $\monolist_v[i] = \texttt{F}$.
The range of $v$ is $[x_1(v), x_2(v)]$ for 
$x_1(v)=x_1(h_v)$ and $x_2(v)=x_2(h_v)$.
There are at most $n+1$ nodes in $\tree$, and each node maintains an array of size 
$O(m)$, so $\tree$ itself uses $O(nm)$ space in total.
See Figure~\ref{fig:planesweep}(a). 

The list $\boundarylist$ consists of 
boundary lists $\boundarylist(v)$ for nodes $v$ of $\tree$.
Each node $v$ of $\tree$ has a pointer to its boundary list $\boundarylist(v)$,
which is a doubly-linked list of the boundary points 
(including the endpoints of $h_v$) lying on $h_v$.
Each boundary point in $\boundarylist$ is the intersection of $L$ 
and a vertical edge of $\rfmpy$, so there are $O(nm)$ boundary points in $\boundarylist$.

Let $d_\delta(p)=d(s,p)$ for a site $s\in f_\delta(p)$ if $f_\delta(p)\neq\emptyset$, or $d_\delta(p)=-\infty$
for $\delta\in\{\py,\ny,\px,\nx\}$.
The list $\distlist$ consists of distance functions $d_v$ for nodes $v$ of $\tree$.
Let $\pl(r)$ denote a point on $L$ with $x(\pl(r)) = r$ for a real number $r$.
Each node $v$ of $\tree$ has a pointer to its
distance function $d_v(x)= d_\py(\pl(x))$ for $x$ in the range $[x_1(v),x_2(v)]$ of $v$.
It is a piecewise linear function with pieces (segments) of slopes $1$ or $-1$.
See Figure~\ref{fig:planesweep}(b).

There are three types of events: (1) a site event, (2) a bottom-side event, and 
(3) a top-side event.
A site event occurs when $L$ encounters a site in $\sset$.
A bottom-side event occurs when $L$ encounters the bottom side of a rectangle in $\rectset$.
A top-side event occurs when $L$ encounters the top side of a rectangle in $\rectset$.
Thus, there are $m$ site events, $n$ bottom-side events, and $n$ top-side events.
See Figure~\ref{fig:events}.
\begin{figure}[t]
  \begin{center}
    \includegraphics[width=\textwidth]{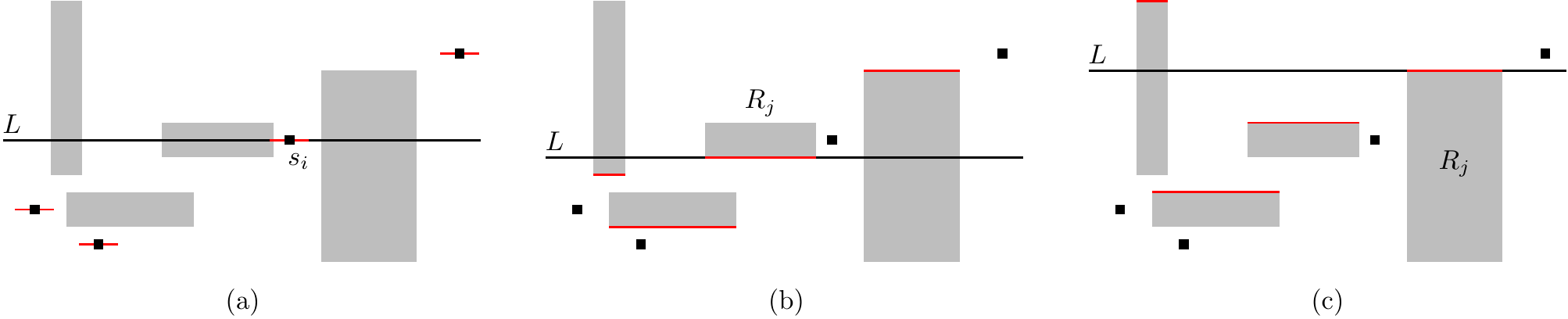}
    \caption{\small
Three types of events.
(a) site events. (b) bottom-side events. (c) top-side events.
     }
    \label{fig:events}
  \end{center}
\end{figure}

We maintain and update $\tree$, $\boundarylist$ and $\distlist$ during the plane sweep for those events.
To handle events, we first sort the events in $y$-coordinate order, which
takes $O((n+m)\log(n+m)) = O(n\log n + m\log m)$ time.
We update $d_v(x)$ only at those events
and keep it unchanged between two consecutive events.
To reflect the distances from sites to $\pl(x)\in h_v$ correctly,
we assign an additive weight to $d_v(x)$, which is the difference in 
the $y$-coordinates between the current event and the last event 
at which $d_v(x)$ is updated.

Initially, when $L$ is at 
the bottom side of $\boundbox$,
$\tree$ consists of one node $v$ with $x_1(v) = x_1(\boundbox)$, $x_2(v) = x_2(\boundbox)$,
and $\monolist_v[i] = \texttt{F}$ for all $i\in\{1,\ldots,m\}$.
$\boundarylist(v)$ has no boundary point and $d_v(x) = -\infty$ for all $x$,
since no points on $L$ is $\py$-reachable from any sites.

\subsection{Handling a site event}
When $L$ encounters a site $s_i\in\sset$, we find the node $v\in\tree$ such that 
$x_1(v) \leq x(s_i) \leq x_2(v)$. 
Every point on $h_v$ is $\py$-reachable from $s_i$,
so we set $\monolist_v[i] = \texttt{T}$.
We can find $v$ in $O(\log n)$ time, and set $\monolist_v[i] = \texttt{T}$ in constant time.
Thus, it takes $O(\log n)$ time to update $\tree$.

For any point $\pl(x)\in h_v$,
$d(s_i,\pl(x)) = |x - x(s_i)|$.
By Lemma~\ref{lem:reverseorder}, there is at most one
maximal interval $I \subset [x_1(v),x_2(v)]$ 
such that $d_v(x)<d(s_i,\pl(x))$ for every $x \in I$.
Moreover, $I$ is bounded from left by $x_1(v)$ or from right by $x_2(v)$
because $d_v(x)$ is continuous and  
consists of pieces (segments) of slopes
$1$ or $-1$, and $d(s_i,\pl(x)) = |x - x(s_i)|$.
We find the boundary point $\pl(x^*)\in h_v$ induced by $s_i$ such that 
$d_v(x^*) = d(s_i,\pl(x^*))$.
If $I$ is bounded from left, we update $d_v(x)$ to $d_v(x)=d(s_i,\pl(x))$ for $x \leq x^*$.
If $I$ is bounded from right, we update $d_v(x)$ to $d_v(x)=d(s_i,\pl(x))$
for $x \geq x^*$.

If there is no such point $\pl(x^*)$,
either $d_v(x)<d(s_i,\pl(x))$ or $d_v(x)>d(s_i,\pl(x))$
for all $x$ with $x_1(v) \leq x \leq x_2(v)$.
If $d_v(x)<d(s_i,\pl(x))$, we update $d_v(x)$ to $d_v(x)=d(s_i,\pl(x))$
for $x_1(v) \leq x \leq x_2(v)$.
If $d_v(x)>d(s_i,\pl(x))$, we do not update $d_v(x)$.

We update $\boundarylist(v)$ by removing all the boundary points 
of $\boundarylist(v)$ lying in the interior of $I$ in time linear to the number of 
the boundary points, and then inserting $\pl(x^*)$ %
into $\boundarylist(v)$. 

Since there are $m$ site events, it takes $O(m \log n)$ time in total to update $\tree$.
The total time to remove the boundary points is linear
to the total number of boundary points in $\dspy$, which is $O(nm)$.

\begin{lemma}\label{time:site}
We can handle all site events in $O(nm)$ time using $O(nm)$ space.
\end{lemma}

\subsection{Handling a bottom-side event}
When $L$ encounters the bottom side of a rectangle $R\in\rectset$,
the line segment of $\freespace \cap L$ incident to
the bottom side is replaced by two line segments by the event. See Figure~\ref{fig:events}(b).
Thus, we update $\tree$ by finding the node $v\in\tree$ with 
$x_1(v) \leq x_1(R) < x_2(R) \leq x_2(v)$, removing $v$ from $\tree$, and then 
inserting two new nodes $u$ and $w$ into $\tree$.
We set $(x_1(u),x_2(u)) = (x_1(v),x_1(R))$,
$(x_1(w),x_2(w)) = (x_2(R),x_2(v))$, $\monolist_u = \monolist_v$,
and $\monolist_w = \monolist_v$.
This takes $O(\log n)$ time 
since $\tree$ is a balanced binary search tree.
It takes $O(m)$ time to copy the Boolean values of $\monolist_v$ to $\monolist_u$ and $\monolist_w$, and to remove $\monolist_v$.
Thus, it takes $O(m + \log n)$ time to update $\tree$.

We update $\boundarylist$ by inserting two lists $\boundarylist(u)$ and 
$\boundarylist(w)$ into $\boundarylist$, 
copying the boundary points of $\boundarylist(v)$ to the lists, 
and then removing $\boundarylist(v)$ from $\boundarylist$.
By Corollary~\ref{co:reverse},
$h_v$ intersects $O(m)$ cells in $\rfmpy$.
Thus, $\boundarylist(v)$ has $O(m)$ boundary points, and 
the update to $\boundarylist(u)$ and $\boundarylist(w)$ takes $O(m)$ time.
There is no change to distance functions. 

Since there are $n$ bottom-side events, it takes $O(nm + n \log n)$ time to update $\tree$
and $O(nm)$ time to update $\boundarylist$ for all bottom-side events.

\begin{lemma}\label{time:bottom}
We can handle all bottom-side events in $O(nm + n \log n)$ time using $O(nm)$ space.
\end{lemma}

\subsection{Handling a top-side event}\label{sec:top}
When $L$ encounters the top side of a rectangle $R\in\rectset$,
the two consecutive segments in $\freespace \cap L$ incident to $R$
are replaced by one segment spanning them by the event. See Figure~\ref{fig:events}(c).
We update $\tree$ by finding the two nodes $u, w\in\tree$ with $x_2(u) = x_1(R)$ and 
$x_1(w)=x_2(R)$, removing $u$ and $w$ from $\tree$, and then inserting 
a new node $v$ into $\tree$.
We set $x_1(v)=x_1(u)$, $x_2(v)=x_2(w)$, and 
$\monolist_v[i] = \monolist_u[i]\vee\monolist_w[i]$ for each $i=1,\ldots,m.$
This takes $O(m + \log n)$ time. 

We update the distance function $d_v(x)$
for $x$ with $x_1(v) \leq x \leq x_1(R)$
as follows.
The geodesic path from any point $\pl(x)\in h_u$  
to $s_i$ with $\monolist_u[i] = \texttt{F}$ and $\monolist_w[i] = \texttt{T}$
is $xy$-monotone by Lemma~\ref{lem:monotone}, and thus 
$d(s_i,\pl(x)) = y(\pl(x)) - y(s_i) + |x(s_i) - x|$.
Also, we observe that $x(s_i) \geq x$ for any $x$.
Thus, every $\pl(x)$ has the same site $s^*$ as its farthest site 
among the sites $s_i$
with $\monolist_u[i] = \texttt{F}$ and $\monolist_w[i] = \texttt{T}$.
Then $d(s^*,\pl(x)) = y(\pl(x)) - y(s^*) + x(s^*) - x$.
By Lemma~\ref{lem:reverseorder}, there is at most one maximal interval $I$
of $x\in [x_1(v),x_1(R)]$ such that $d_v(x) \leq d(s^*,\pl(x))$.
Moreover, $I$ is bounded from left by $x_1(v)$.
We find the boundary point $\pl(x^*)\in h_u$ such that $d_v(x^*) = d(s^*,\pl(x^*))$,
and update $d_v(x)$ to $d(s^*,\pl(x))$ for $x \leq x^*$.

If there is no such point $\pl(x^*)$,
either $d_v(x)<d(s^*,\pl(x))$ or $d_v(x)>d(s^*,\pl(x))$
for all $x$ with $x_1(v) \leq x \leq x_1(R)$.
If $d_v(x)<d(s^*,\pl(x))$, we update $d_v(x)$ to $d_v(x)=d(s^*,\pl(x))$
for $x_1(v) \leq x \leq x_1(R)$.
If $d_v(x)>d(s^*,\pl(x))$, we do not update $d_v(x)$.

We update $\boundarylist[x_1(v),x_1(R)]$,
which is a part of $\boundarylist(v)$ with range $[x_1(v), x_1(R)]$,
by removing all the
boundary points in the interior of $I$ in time linear to the number of 
the boundary points, and then inserting $\pl(x^*)$ as a boundary point.
We can handle the case of $x$ with $x_2(R) \leq x \leq x_2(v)$,
and update $\boundarylist[x_2(R),x_2(v)]$ analogously.

\subsubsection{Computing distance functions for a top side}
We show how to compute $d_v(x)$ for $x\in[x_1(R), x_2(R)]$
and update $\boundarylist[x_1(R),x_2(R)]$ efficiently.
Lemma~\ref{lem:monotone} implies the following observation.

\begin{observation}\label{ob:shortest_on_top}
For any point $p$ on the top side of a rectangle $R\in\rectset$ and any site $s\in\sset$ from which $p$ is $\py$-reachable,
every geodesic path from $p$ to $s$ passes through the top-left corner or 
the top-right corner of $R$.
\end{observation}

For an index $k$,
let $\alpha_k$ and $\beta_k$ denote 
the top-left corner and the top-right corner of $R_k\in\rectset$, and
let $\sset_k$ denote the set of the sites
that lie below the polygonal curve consisting of
$\pidl(\alpha_k)$, the top side of $R_k$, and $\pidr(\beta_k)$.

For the top-side event of $R=R_k$, let $\alpha=\alpha_k$ and  $\beta=\beta_k$. 
Note that $x(\alpha) = x_1(R)$ and $x(\beta) = x_2(R)$.
Let $\sset^T$ be the set of the sites $s_i$, with 
$\monolist_v[i] = \texttt{T}$ for all $i=1,\ldots,m$.
We partition $\sset^T$ into three disjoint subsets, $\sset_k$, 
$\sset(\alpha)$, and $\sset(\beta)$,
such that $\sset(\alpha) = \{s_i \in \sset^T\setminus\sset_k \mid x(s_i) \leq x_1(R)\}$ and
$\sset(\beta) = \{s_i \in \sset^T\setminus\sset_k \mid x(s_i) \geq x_2(R)\}$.
See Figure~\ref{fig:topside}.

Every geodesic path from any site in $\sset(\alpha)$ or $\sset(\beta)$ to any point on 
the top side of $R$ is $xy$-monotone. Thus for any point $\pl(x)$ 
lying on the top side of $R$, 
we can compute $d(s^\alpha, \pl(x))$ and $d(s^\beta, \pl(x))$, where
$s^\alpha$ and $s^\beta$ are the farthest sites of $\pl(x)$ among sites in
$\sset(\alpha)$ and among sites in $\sset(\beta)$, respectively, 
as we did 
for $\boundarylist[x_1(v),x_1(R)]$ or $\boundarylist[x_2(R),x_2(v)]$.

Let $R_a$ be the rectangle hit first by the vertical ray emanating from $\alpha$ going downwards,
and let $R_b$ be the rectangle hit first by the vertical ray emanating from $\beta$ going downwards.
See Figure~\ref{fig:topside}.

We compute $d(\alpha, s)$ and $d(\beta,s)$ for each $s\in\sset_k$,
and then compute $d_v(x)$ for $x$ with $x_1(R) \leq x \leq x_2(R)$,
where $v$ is the node of $\tree$ corresponding to the top-side event of $R$.
The top-side events by $R_a$ and $R_b$ were handled before the top-side event of $R$,
and thus we have $d(\alpha_a, s)$ and $d(\beta_a,s)$ for sites $s\in\sset_a$,
and $d(\alpha_b, s')$ and $d(\beta_b,s')$ for sites $s'\in\sset_b$.
By Observation~\ref{ob:shortest_on_top}, we can compute $d(\alpha, s)$ and $d(\beta, s)$ for a site $s\in\sset_k$ 
as follows.
\begin{itemize}
\item
If $s\in\sset_a$, 
$d(\alpha,s) = \min\{d(\alpha,\alpha_a) + d(\alpha_a,s), d(\alpha,\beta_a) + d(\beta_a, s) \}$.
If $s\notin\sset_a$,
$d(\alpha,s) = |x(\alpha) - x(s)| + |y(\alpha) - y(s)|$
since $\pi(\alpha,s)$ is $xy$-monotone by Lemma~\ref{lem:monotone}.
\item
If $s\in\sset_b$, $d(\beta,s) = \min\{d(\beta,\alpha_b) + d(\alpha_b,s), d(\beta,\beta_b) + d(\beta_b, s) \}$.
If $s\notin\sset_b$,
$d(\beta,s) = |x(\beta) - x(s)| + |y(\beta) - y(s)|$
since $\pi(\beta,s)$ is $xy$-monotone by Lemma~\ref{lem:monotone}.
\end{itemize}

\begin{figure}[t]
  \begin{center}
    \includegraphics[width=.5\textwidth]{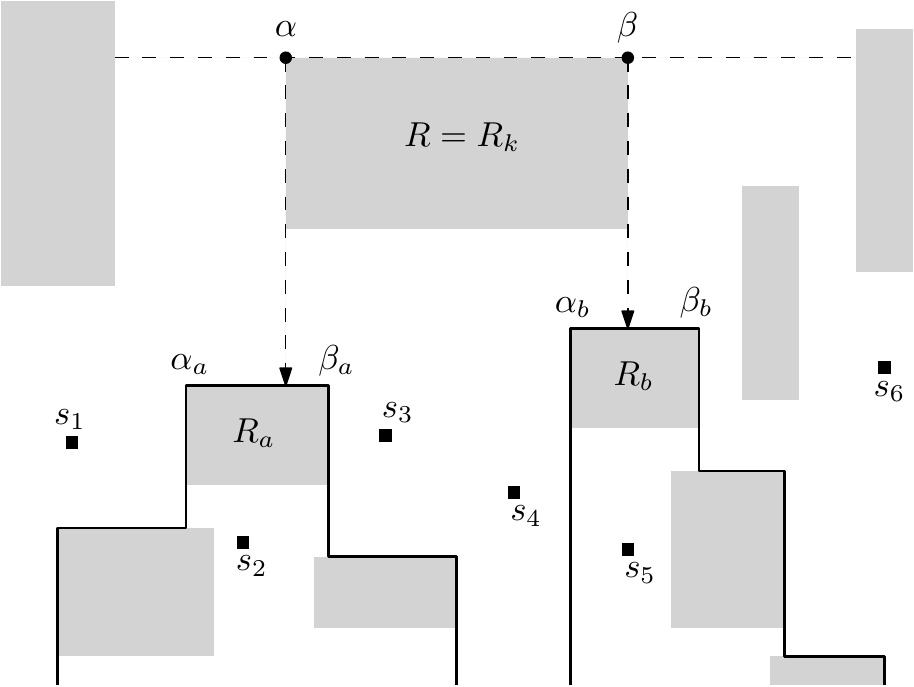}
    \caption{\small
$\sset^T = \{s_1,s_2,s_3,s_4,s_5,s_6\}$ is partitioned into 
$\sset_k = \{s_2,s_3,s_4,s_5\}$, $\sset(\alpha) = \{s_1\}$, and $\sset(\beta) = \{s_6\}$.
For two rectangles $R_a$ and $R_b$,
$\sset_a = \{s_2\}$ and $\sset_b = \{s_5\}$.
     }
    \label{fig:topside}
  \end{center}
\end{figure}

By Observation~\ref{ob:shortest_on_top},
every geodesic path from $s$ to $\pl(x)$ passes through either $\alpha$ or $\beta$.
We denote by $\da(i,x) = d(\alpha,s_i) + x - x(\alpha)$ the length of a geodesic path from a site $s_i$ to $\pl(x)$ passing through $\alpha$,
and denote by $\db(i,x) = d(\beta,s_i)+x(\beta)-x$ the length of a geodesic path from $s_i$ to $\pl(x)$ passing through $\beta$.
Let $D(x) = \max_{s_i\in\sset_k}\min\{\da(i,x),\db(i,x)\}$
for all $x$ with $x(\alpha) \leq x \leq x(\beta)$.
For ease of description, let $\da(i) = d(\alpha,s_i)$
and $\db(i) = d(\beta,s_i) + x(\beta) - x(\alpha)$.
We observe that $\da(i) \leq \db(i)$ for every $s_i\in\sset_k$.
Let $\dba(i)=\db(i)-\da(i)$.

Our goal is to compute $D(x)$ in $O(m)$ time.
To achieve this,
we consider two cases, either (1) $\dba(a) \geq \dba(b)$ 
for every indices $a$ and $b$ with $1\leq a < b\leq m$,
or (2) $\dba(a) < \dba(b)$ for some indices $a$ and $b$ with $1\leq a < b\leq m$.
The following two lemmas show how to compute $D(x)$ in $O(m)$ time
for these two cases.

\begin{figure}[ht]
  \begin{center}
    \includegraphics[width=\textwidth]{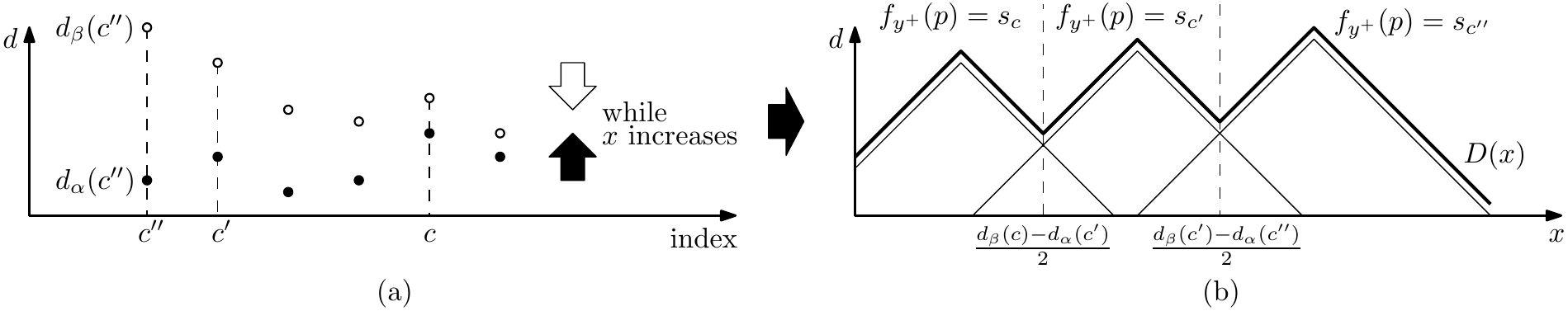}
    \caption{\small
Two graphs showing the distances from sites to $\pl(x)$ for $x_1(R) \leq x\leq x_2(R)$ 
with different domains.
We find sites $s_c = \max_{i\in\{1,\ldots,m\}} d_\alpha(i)$,
$s_{c'} = \max_{i\in\{1,\ldots,c-1\}} d_\alpha(i)$, and $s_{c''} = \max_{i\in\{1,\ldots,c'-1\}} \da(i)$ recursively.
They are the farthest sites from $\pl(x)$ moving rightwards.
(a) The graph with respect to indices of sites.
(b) The graph for $x$ with $x_1(R) \leq x \leq x_2(R)$.
     }
    \label{fig:topside-graph}
  \end{center}
\end{figure}

\begin{lemma}\label{lem:diffmonotone}
If $\dba(a) \geq \dba(b)$ for every two indices $a$ and $b$ 
with $1\leq a < b\leq m$, we can compute $D(x)$ in $O(m)$ time.
\end{lemma}
\begin{proof}
Let $s_c\in f_\py(\alpha)$ be the farthest site from $\alpha$ with the smallest index $c$.
By Lemma~\ref{lem:reverseorder} and $x(s_a)\le x(s_b)$ for every two indices $a$ and $b$ with $1\leq a < b\leq m$,
no site $s_i \not\in f_\py(\alpha)$ for any $i=c+1,\ldots,m$ is in $f_\py(\pl(x))$ for any point $\pl(x)$ on the top side of $R$. 
Observe that $\da(c) + t \leq \db(c) - t$
for $0\leq t\leq \dba(c) / 2$.
Also, $\da(i) + t \leq \db(i) - t$ still holds
for $0\leq t\leq \dba(c) / 2$
since $\dba(i) \geq \dba(c)$
for every index $i=1,\ldots,c-1$.

Let $c'$ be the smallest index satisfying
$\da(c') = \max_{i\in\{1,\ldots,c-1\}}\da(i)$.
Then $s_{c'}$ is in $f_\py(\pl(x(\alpha)+t))$ for $t \geq (\db(c) - \da(c')) / 2$.
For $\dba(c) / 2 \leq t \leq (\db(c) - \da(c')) / 2$,
$s_c \in f_\py(\pl(x(\alpha)+t))$.
Therefore, $\pl((\db(c) - \da(c')) / 2)$ becomes a boundary point.
See Figure~\ref{fig:topside-graph}.
Using $c'$, we compute the smallest index $c''$ satisfying
$\da(c') = \max_{i\in\{1,\ldots,c'-1\}}\da(i)$
and find all farthest sites $f_\py(\pl(x(\alpha)+t))$ for all $t$ recursively.

We use a stack storing indices of sites to compute $c'$ and $c''$ recursively.
We can find $s_c$ in $O(m)$ time. 
Let $\stacktop$ be the top element (site) of the stack.
Initially, the stack contains $c$.
Let $j$ be the index at the $i$-th iteration for $i$ from $c-1$ to $1$.
We pop $\stacktop$ from the stack until $\da(j) \geq \da(\stacktop)$.
Then we push $j$ into the stack if $\da(j) < \da(\stacktop)$.
Observe that the stack never be empty since $\da(c) > \da(j)$ for all $j\in\{1,\ldots,c-1\}$.

The pseudocode of the algorithm
is given in Algorithm~\ref{al:farthest}.
We repeat this until we compute $D(x)$ and all boundary points.
This takes $O(m)$ time.
\end{proof}

\begin{algorithm}\caption{Computing Farthest Sites on Rectangle}\label{al:farthest}
\begin{algorithmic}
\Procedure{FarthestSites}{$\{s_1,\cdots,s_m\},R$} \Comment{$x(s_1) \le \ldots \le x(s_m)$}
	\State $c \gets \argmax_{i\in\{1,\ldots,m\} }\da(i)$ \Comment{$\alpha$ is the top-left corner of $R$}
	\State stack $A \gets \{c\}$ 
	\For{$i = c-1$ to $1$} 
		\If{$\da(i) < \db(\stacktop(A))$} \Comment{$\beta$ is the top-right corner of $R$}
			\State push $i$ into $A$
		\ElsIf{$\da(i) \geq \da(\stacktop(A))$}
			\While{$\da(i) \geq \da(\stacktop(A))$}
				\State pop $\stacktop(A)$ from $A$
			\EndWhile
			\State push $i$ into $A$
		\EndIf
	\EndFor \Comment{$A\neq\emptyset$ during all iterations by $c$.}
	\State sort $A$ in reverse order 
	\State $x \gets x(\alpha)$
	\State set $\pl(x)$ as a boundary point
	\While{$|A| > 1$}
		\State $i \gets \stacktop(A)$
		\State pop $\stacktop(A)$ from $A$
		\State $x' \gets (\db(i)-\da(\stacktop(A)))/2$
		\State $D(t) \gets \min\{d(\alpha,s_i) + t - x(\alpha),d(\beta,s_i)+x(\beta)-t\}$ for $x \leq t \leq x'$
		\State $x \gets x'$
		\State set $\pl(x)$ as a boundary point
	\EndWhile
	\State $x' \gets x(\beta)$
	\State $D(t) \gets \min\{d(\alpha,s_{\stacktop(A)}) + t - x(\alpha),d(\beta,s_{\stacktop(A)})+x(\beta)-t\}$ for $x \leq t \leq x'$
	\State set $\pl(x')$ as a boundary point
	\State \textbf{return} $D(x)$ and boundary points
\EndProcedure
\end{algorithmic}
\end{algorithm}

If $d_{\beta\alpha}(a) < d_{\beta\alpha}(b)$ for some indices $a$ and $b$ with $a<b$, 
we can remove either $s_a$ or $s_b$ by the following lemma.

\begin{lemma}\label{lem:diffprune}
If there are two indices $a$ and $b$ with $a<b$ such that  
$d_{\beta\alpha}(a) < d_{\beta\alpha}(b)$,
either $s_a$ or $s_b$
is a farthest site from no point $\pl(x)$ for $x$ with $x_1(R)\leq x\leq x_2(R)$.
\end{lemma}
\begin{proof}
The proof is similar to that of Lemma~\ref{lem:reverseorder}.
First, we show that
if there are two geodesic paths $\pi(\alpha,s_b)$ and $\pi(\beta,s_a)$
that intersect each other, then $d_{\beta\alpha}(a) \geq d_{\beta\alpha}(b)$.
Let $q$ be a point in the intersection of the paths.
Then 
\begin{equation*}
	d(\alpha,s_b) = d(\alpha,q)+d(q,s_b) \quad\text{and}\quad d(\beta,s_a) = d(\beta,q)+d(q,s_a).	
\end{equation*}

We observe that $d(\alpha,s_a) \leq d(\alpha,q)+d(q,s_a)$ and $d(\beta,s_b) \leq d(\beta,q)+d(q,s_b)$.
Adding these two inequalities,
we obtain 
\begin{eqnarray*}
	d(\alpha,s_a) + d(\beta,s_b) &\leq& d(\alpha,q)+d(q,s_b)+d(\beta,q)+d(q,s_a)\\
	&=& d(\alpha,s_b) + d(\beta,s_a)
\end{eqnarray*}
and thus, $d(\beta,s_b) - d(\alpha,s_b) \leq d(\beta,s_a) - d(\alpha,s_a)$.
Since $d(\beta,s_b) - d(\alpha,s_b) = \db(b) + C - \da(b)$ and $d(\beta,s_a) - d(\alpha,s_a) = \db(a) + C - \da(a)$,
we have $\db(b) - \da(b) \leq \db(a) - \da(a)$, where $C = x(\beta) - x(\alpha)$.
Therefore, $\dba(a) \geq \dba(b)$.

By contraposition, if $\dba(a) < \dba(b)$,
no two geodesic paths $\pi(\alpha,s_b)$ and $\pi(\beta,s_a)$ intersect each other.
Since $y(\alpha) = y(\beta)$, $x(s_a) \le x(s_b)$,
and all geodesic paths from the sites we consider are $\py$-monotone,
$y(s_a) \neq y(s_b)$.
We show that $s_b$ is a farthest site from no point $\pl(x)$
if $y(s_a) < y(s_b)$, and $s_a$ is a farthest site from no point $\pl(x)$
if $y(s_a) > y(s_b)$.

Consider the case that $y(s_a) < y(s_b)$.
Then $\piru(s_b)$ intersects $\pi(\beta,s_a)$ at a point, say $q$.
Since $x(s_a) \le x(s_b) \leq x(q)$, $y(s_a) < y(s_b) \leq y(q)$,
and $\pi(q,s_b)$ is $xy$-monotone,
$d(q,s_a) > d(q,s_b)$.
Then 
\begin{equation*}
	d(\beta,s_b) \leq d(\beta,q)+d(q,s_b) < d(\beta,q)+d(q,s_a) = d(\beta,s_a)\quad \text{and} \quad d(\beta,s_b) - d(\beta,s_a) < 0.	
\end{equation*}
Then $\db(b) - \db(a) < 0$ because $\db(b) = d(\beta,s_b) + C$ and $\db(a) = d(\beta,s_a) + C$, where $C = x(\beta) - x(\alpha)$.
Since $\dba(a) < \dba(b)$, we have $\da(b) - \da(a) < \db(b) - \db(a) < 0$,
$\da(b) < \da(a)$ and $\db(b) < \db(a)$. 
This implies that 
\begin{equation*}
	\min\{\da(b)+t,\db(b)-t\} < \min\{\da(a)+t,\db(a)-t\} \quad \text{for any } t\in[0,x(\beta)-x(\alpha)].
\end{equation*}
Thus, $s_b$ is a farthest site from no point $\pl(x)$.

Now consider the case that $y(s_a) > y(s_b)$.
Then $\pilu(s_a)$ intersects $\pi(\alpha,s_b)$ at a point, say $q'$.
Since $x(q') \le x(s_a) \le x(s_b)$, $y(q') \geq y(s_a) > y(s_b)$,
and $\pi(q',s_a)$ is $xy$-monotone,
$d(q',s_b) > d(q',s_a)$.
Thus,
\begin{equation*}
	\da(a) \leq d(\alpha,q')+d(q',s_a) < d(\alpha,q')+d(q',s_b) = \da(b) \quad \text{and} \quad \da(b) - \da(a) > 0.
\end{equation*}
Since $\dba(a) < \dba(b)$,
we have $0 < \da(b) - \da(a) < \db(b) - \db(a)$.
Then $\db(b) > \db(a)$ and $\da(b) > \da(a)$ which implies that 
\begin{equation*}
	\min\{\da(b)+t,\db(b)-t\} > \min\{\da(a)+t,\db(a)-t\} \quad \text{for any } t\in[0,x(\beta)-x(\alpha)].	
\end{equation*}
Therefore, $s_a$ is a farthest site from no point $\pl(x)$.
\end{proof}

\begin{algorithm}\caption{Pruning Sites on Rectangle}\label{al:prune}
\begin{algorithmic}
\Procedure{PruningSites}{$\{s_1,\ldots,s_m\},R$} \Comment{$x(s_1) \le \ldots \le x(s_m)$}
	\State stack $A \gets \emptyset$ \Comment{$\alpha$ is the top-left corner of $R$}
	\For{$i = 1$ to $m$} \Comment{$\beta$ is the top-right corner of $R$}
		\While{\texttt{True}}
			\If{$A=\emptyset$ or $d(\beta,s_i) - d(\alpha,s_i) \geq d(\beta,\stacktop(A)) - d(\alpha,\stacktop(A))$} 
				\State push $s_i$ into $A$
				\State \textbf{break}
			\ElsIf{$d(\alpha,s_i) \leq d(\alpha,\stacktop(A))$}
				\State ignore $s_i$
				\State \textbf{break}
			\ElsIf{$d(\alpha,s_i) > d(\alpha,\stacktop(A))$}
				\State pop $\stacktop(A)$ from $A$
			\EndIf
		\EndWhile
	\EndFor
   	\State \textbf{call} FarthestSites($A,R$)
\EndProcedure
\end{algorithmic}
\end{algorithm}

Lemmas~\ref{lem:diffmonotone} and~\ref{lem:diffprune} imply that the complexity of $D(x)$ is $O(m)$.
By Lemma~\ref{lem:diffprune},
we can remove the sites which never be the farthest sites
by comparing $\dba(a)$ and $\dba(b)$ 
for two indices $a$ and $b$.
This can be done in $O(m)$ time by Algorithm~\ref{al:prune}. 
After pruning, $\dba(a) \geq \dba(b)$
for every pair of remaining sites $s_a$ and $s_b$ with $a < b$.
Therefore, we can compute $D(x)$ in $O(m)$ time by Lemma~\ref{lem:diffmonotone}.
Then we can compute $d_v(x) = \max\{d(s^\alpha, \pl(x)),D(x),d(s^\beta, \pl(x))\}$.
in $O(m)$ time.
We update $\boundarylist[x_1(R),x_2(R)]$ in $O(m)$ time using $d_v(x)$.

There are $n$ top-side events, so we can handle the top-side events in $O(nm + n \log n)$ time.
In addition, we compute distances from $O(m)$ sites to each corner of $O(n)$ rectangles, and store them.
Using ray shooting queries emanating from the corners of rectangles, it takes $O(nm) + O(n \log n)$ time using $O(nm)$ space.
Therefore, we have the following lemma.

\begin{lemma}\label{time:top}
We can handle all top-side events in $O(nm + n \log n)$ time using $O(nm)$ space.
\end{lemma}

\subsection{Constructing the query data structure}
Initially, $\dspy=\emptyset$.
For each site event and top-side event, we update $d_v(x)$ and $\boundarylist(v)$
for node $v$ of $\tree$ corresponding to the event.
We insert a horizontal segment $h$ corresponding to each interval
which is updated at the event into $\dspy$,
and copy the boundary points into $h$.
For each site event, at most one horizontal line segment $h$ is inserted.
There is no boundary point in the interior of $h$,
so we can copy $h$ with two endpoints in $O(1)$ time.
For each top-side event, at most three horizontal line segments are inserted.
They have $O(m)$ boundary points by Lemma~\ref{lem:reverseorder},
so we can copy them in $O(m)$ time.
There are $O(n+m)$ horizontal segments 
and $O(nm)$ boundary points in $\dspy$,
so the query structure $\dspy$ uses $O(nm)$ space.

\subsubsection{Farthest-point queries}
Once $\dspy$ is constructed, 
we can find $f_\py(q)$ from a query point $q\in\freespacenoempty$.
We find the farthest sites from $q$ in the other three maps using 
their query data structures. 

By Corollary~\ref{co:boundary},
our query problem reduces to the vertical ray shooting queries.
We use the data structure by Giora and Kaplan~\cite{giora2009} 
for vertical ray shooting queries on $O(n+m)$ horizontal line segments in $\dspy$,
which requires $O((n+m) \log (n+m))$ time and $O(n+m)$ space for construction.
Let $h$ be the horizontal segment in $\dspy$ hit first
by the vertical ray emanating from $q$ going downwards.
We can find $h$ in $O(\log (n+m))$ time using the ray shooting structure.
If no horizontal segment in $\dspy$ is hit by the ray,
$q$ is $\py$-reachable from no site.
Otherwise, there are $O(m)$ boundary points on $h$, sorted in increasing order of 
$x$-coordinate. With those boundary points, we can find $f(q)$ for a query point $q$ 
in $O(\log m)$ time using binary search. Thus, a farthest-neighbor query takes 
$O(\log(n+m))$ time in total.

Once the farthest sites of $q$ for each of the four data structures is found, 
we take the sites with the largest distance among them as the farthest sites $f(q)$
of $\sset$ from $q$.
Combining Lemmas~\ref{time:site},~\ref{time:bottom} and~\ref{time:top} with query time,
we have the following theorem.

\begin{theorem}
We can construct a data structure for $m$ point sites in the presence of 
$n$ axis-aligned rectangular obstacles in the plane
in $O(nm + n \log n + m \log m)$ time and $O(nm)$ space that 
answers any $L_1$ farthest-neighbor query in $O(\log (n+m))$ time.
\end{theorem}

\section{Computing the Explicit Farthest-point Voronoi Diagram}
\label{sec:comp.fvd}
We construct the explicit farthest-point Voronoi diagram $\fvd=\fvd(\sset,\rectset)$
of a set $\sset$ of $m$ point sites in the presence of 
a set $\rectset$ of $n$ rectangular obstacles in the plane.
It is known that $\fvd$ requires $\Omega(nm)$ space~\cite{bae2009,moshe2001}.
It takes $\Omega(n \log n)$ time
to compute the geodesic distance between two points in $\freespace$~\cite{rezende1985}.
By a reduction from the sorting problem, it can be shown to take 
$\Omega(m \log m)$ time for computing 
the farthest-point Voronoi diagram of $m$ point sites in the plane. 
We present an $O(nm+n\log n+m\log m)$-time algorithm using $O(nm)$ space
that matches the time and space lower bounds.
This is the first optimal algorithm for constructing 
the farthest-point Voronoi diagram of points in the presence of obstacles
in the plane in both time and space. 

We construct $\dspy$ using the plane sweep in Section~\ref{sec:algorithms}.
During the plane sweep, we find all horizontal edges of $\rfmpy$ 
and insert them into $\dspy$ as segments.
We find all the lower endpoints of 
the vertical edges of $\rfmpy$ and insert them as boundary points in 
$\boundarylist$.
We also find the upper endpoints of vertical 
edges of $\rfmpy$. By connecting those endpoints using 
vertical segments appropriately, we can construct $\rfmpy$ from $\dspy$
in a doubly connected edge list without increasing the time and space complexities.
The other three maps can also be constructed in the same way
in the same time and space.

We construct the farthest-point Voronoi diagram $\fvd$ using the four maps explicitly.
Note that $f(p)=f_{\py}(p)$ for any point $p$ lying on the top side of $\boundbox$.
Thus, it suffices to compute $\fvd$ in $\freespace\cap\boundbox$.
For ease of description, we assume that the $x$-coordinates of the rectangles in $\rectset$ are all distinct.
We consider a vertical decomposition $\freespace_V$ 
obtained by drawing maximal vertical line segments contained in $\freespace\cap\boundbox$
of which each is extended from a vertical side of a hole of $\freespace$.
Let $V$ be a set of such vertical line segments.
$\freespace \setminus \bigcup_{\ell\in V}\ell$ consists of $O(n)$ connected faces. 
Each face is a rectangle since each hole of $\freespace$ is a rectangle and 
$\freespace$ is bounded by $\boundbox$.
See Figure~\ref{fig:sectionfvd}(a).

\begin{figure}[t]
  \begin{center}
    \includegraphics[width=\textwidth]{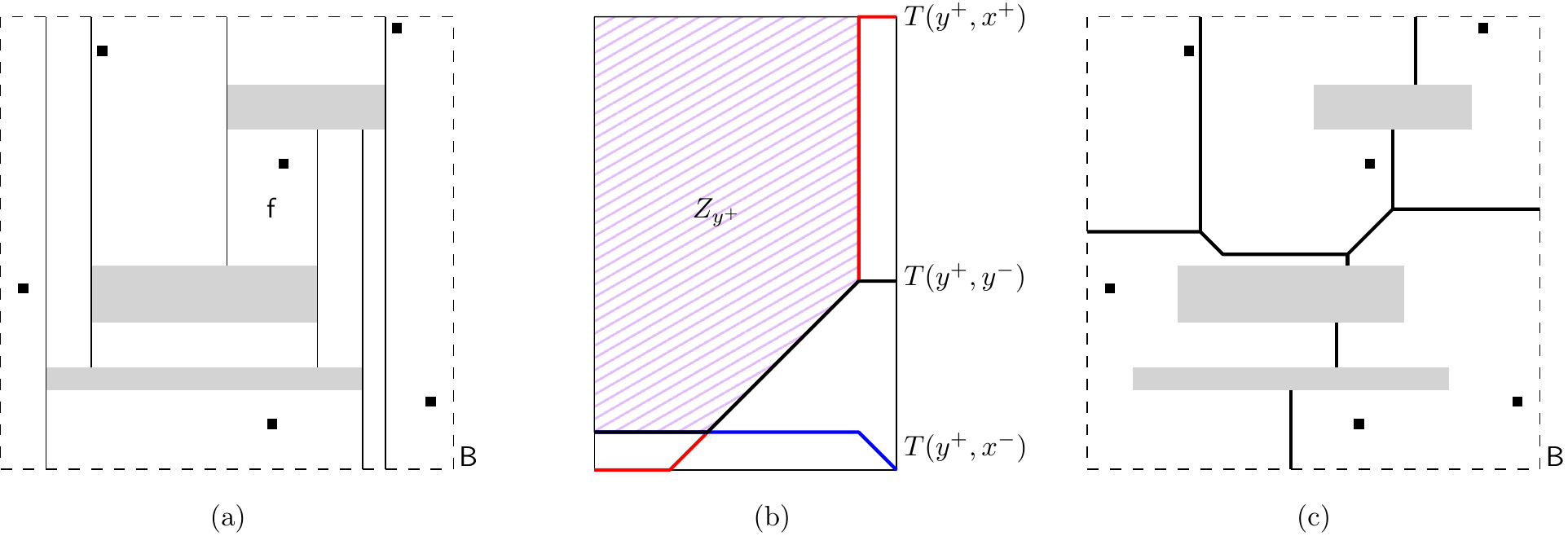}
    \caption{\small
(a) Vertical decomposition $\freespace_V$. $\face$ is a face of $\freespace_V$.
(b) $Z_\py$ is a region in $\face$ above the upper envelope of three traces, $\trace(\py,\ny)$, $\trace(\py,\px)$ and $\trace(\py,\nx)$.
(c) Explicit geodesic $L_1$ farthest-point Voronoi diagram $\fvd$.
}
    \label{fig:sectionfvd}
  \end{center}
\end{figure}

Any two farthest-point maps $\rfm_1, \rfm_2$ have a \emph{bisector} 
which consists of the points in $\freespace$ having the same distance
to their farthest sites in $\rfm_1$ and in $\rfm_2$.
The four maps define six bisectors.
In a face of $\freespace_V$, 
the six bisectors and some axis-aligned segments partition the face into 
\emph{zones} such that $\fvd$ restricted to one zone coincides with the diagram 
in the corresponding region of a farthest-point map.
Thus, we compute the bisectors between maps in each face of $\freespace_V$, 
partition the face into zones,
find the region of a farthest-point map corresponding to each zone,
and then glue the regions and faces to compute $\fvd$ completely.

\subsection{Bisectors of farthest-point maps}
We define the \emph{bisector} between 
$\rfm_\delta$ and $\rfm_{\delta'}$ as 
$\mbisector(\delta,\delta') = \{q\in\freespace \mid d_\delta(q) = d_{\delta'}(q)\}$
for any two distinct $\delta,\delta'\in\{\py,\ny,\px,\nx\}$.
We show some structural and combinatorial properties of the bisectors
between two farthest-point maps. 

\begin{figure}[t]
  \begin{center}
    \includegraphics[width=.5\textwidth]{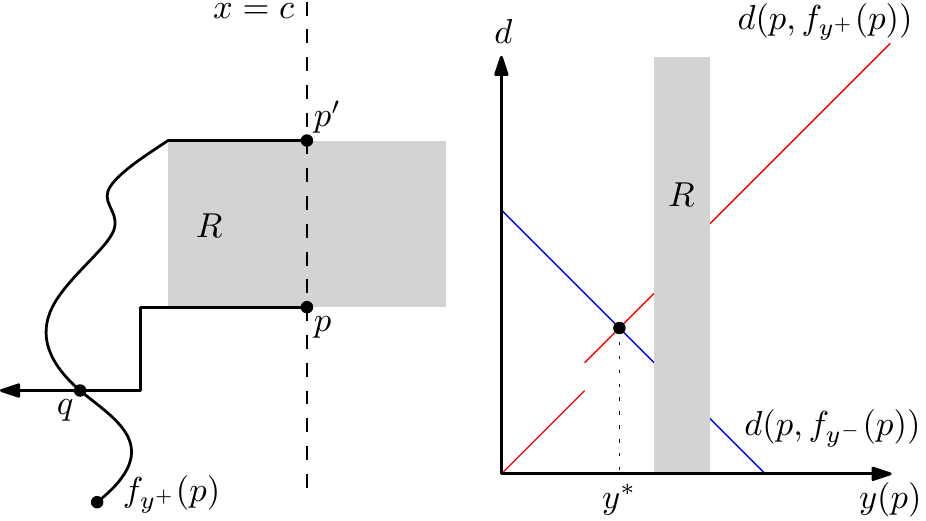}
    \caption{\small
Proof of Lemma~\ref{lem:mapbisect_y}.
$d(p,f_\py(p))$ can be discontinuous, but it is increasing.
$d(p,f_\ny(p))$ can be discontinuous, but it is decreasing.
At $p^*$ with $y(p^*) = y^*$, $d(p^*,f_\py(p^*)) = d(p^*,f_\ny(p^*))$ and $p^*$ is in $\mbisector(\py,\ny)$.
     }
    \label{fig:lemmaforfvd}
  \end{center}
\end{figure}

\begin{lemma}\label{lem:mapbisect_y}
Any vertical line intersects $\mbisector(\py,\ny)$ in at most one point.
\end{lemma}
\begin{proof}
Imagine that a point $p$ moves vertically upwards.
We show that $d_\py(p)$ increases as $p$ moves.
By Corollary~\ref{co:boundary}, $f_\py(p)$ does not change
until $p$ meets a horizontal edge of $\rfmpy$, 
or a rectangle of $\rectset$, and thus
$d_\py(p)$ increases as $p$ moves.
When $p$ meets a horizontal edge of $\rfmpy$,
$f_\py(p)$ changes, but $d_\py(p)$ still increases
by the definition of $f_\py(p)$.

Consider the case that $p$ meets the bottom side of a rectangle $R\in\rectset$.
See Figure~\ref{fig:lemmaforfvd}.
Let $p'$ be the point on the top side of $R$ with $x(p') = x(p)$.
Then there is a geodesic path $\pi(p',f_\py(p))$
that passes through the top-left (or the top-right) corner of $R$.
Without loss of generality, assume it passes the top-left corner of $R$.
By the definition of $f_\py(p)$, $\pild(p)$ intersects $\pi(p',f_\py(p))$ at a point, say $q$.
Then, \[d(p,f_\py(p)) \leq d(p,q) + d(q,f_\py(p)) < d(p',f_\py(p)) \leq d(p',f_\py(p')).\]
Therefore, $d(p,f_\py(p))$ still increases as $p$ jumps to $p'$.
Likewise, $d(p,f_\ny(p))$ decreases as $p$ moves vertically upwards.
Then $d(p,f_\py(p))=d(p,f_\ny(p))$ occurs at most once at a moment
for $p$ moving vertically upwards.
Thus, any vertical line intersects $\mbisector(\py,\ny)$ in at most one point.
\end{proof}

By Lemma~\ref{lem:mapbisect_y},
$\mbisector(\py,\ny)$ is $x$-monotone consisting of
segments of slopes $0$, $+1$, or $-1$.

\begin{lemma}\label{lem:mapbisect_xy}
For any vertical line segment $\ell$ contained in $\freespace$, $\ell\cap\mbisector(\py,\px)$ consists of at most one connected component.
\end{lemma}
\begin{proof}
The proof is similar to the one for Lemma~\ref{lem:mapbisect_y}.
Consider a vertical line segment $\ell$ contained in $\freespace$.
Imagine that a point $p$ moves vertically upwards from the lower endpoint of $\ell$ to the upper endpoint.
We show that $d_\py(p)-y(p)=d_1(p)$ does not decrease.
By Corollary~\ref{co:boundary}, $f_\py(p)$ does not change
until $p$ meets a horizontal edge of $\rfmpy$, 
or a rectangle of $\rectset$, and thus
$d_1(p)$ remains the same as $p$ moves.
When $p$ meets a horizontal edge of $\rfmpy$, $f_\py(p)$ changes and
$d_1(p)$ increases by the definition of $f_\py(p)$.

We also show that $d_\px(p)-y(p)=d_2(p)$ does not increase.
Recall that every distance function in $\distlist$ consists of 
pieces of slopes $1$ or $-1$ during the plane sweep in Section~\ref{sec:algorithms}.
Therefore, $d_2(p)$ does not increase. Thus there is at most one connected component satisfying $d_\py(p)=d_\px(p)$ in $\ell$.
\end{proof}

We can also show that Lemma~\ref{lem:mapbisect_xy} holds for $\mbisector(\py,\nx)$.
See Figure~\ref{fig:mbisect} for bisectors. Lemmas~\ref{lem:mapbisect_y} and~\ref{lem:mapbisect_xy} imply that $f(p')=f_\py(p')$ if $p'p$ is a vertical line segment contained in $\freespace$ with $y(p')>y(p)$, and $f(p)=f_\py(p)$. 
Thus, these three bisectors contained in a face of $\freespace_V$ are 
$x$-monotone.

\begin{figure}
	\begin{center}
	  \includegraphics[width=\textwidth]{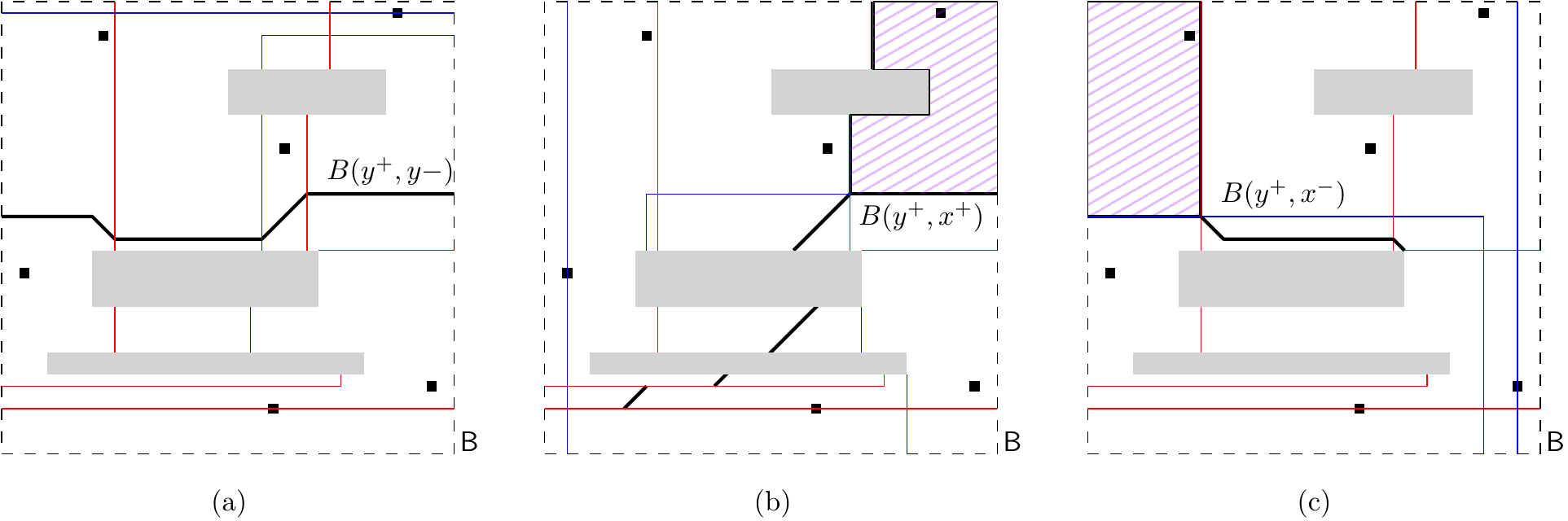}
	  \caption{\small
  Three bisectors of $\rfmpy$ (red) and the other three farthest-point maps (blue).
  (a) $\mbisector(\py,\ny)$. (b) $\mbisector(\py,\px)$. (c) $\mbisector(\py,\nx)$. Bisectors may contain a two-dimensional region.}
	  \label{fig:mbisect}
	\end{center}
  \end{figure}

For each face $\face$ of $\freespace_V$, we compute the portion of $\mbisector(\py,\ny)$ contained in $\face$.
As $\mbisector(\py,\ny)\cap\face$ is $x$-monotone, 
we sweep a vertical line $L$ from $x_1(\face)$ to $x_2(\face)$ maintaining a point $p\in \face\cap L$ with $d_\py(p) = d_\ny(p)$.
First, we compute $p$ lying on the left side of $\face$ as follows.
There are $O(m)$ intersections of the left side of $\face$ with the horizontal segments of $\dspy$ and $\dsny$
as any vertical line segment contained in $\freespace$ intersects $O(m)$ horizontal segments of them.
For each intersection point $q$, 
we compute $d_\py(q)$ and $d_\ny(q)$,
and find two consecutive points $q_1$ and $q_2$ among the intersection points by $y$-coordinate
such that $d_\py(q_1) \leq d_\ny(q_1)$ and $d_\ny(q_2) \leq d_\py(q_2)$.
We can compute $q_1$ and $q_2$ in $O(m)$ time
using $\dspy$ and $\dsny$. 
Then we compute $p$ lying on $q_1q_2$.

Having the distance functions, we have the slope of the bisector incident to $p$.
Let $\vec{\ell}$ be the half-line from $p$ with the slope going rightward. 
We find the first point $p'$ on $\vec{\ell}$ from $p$ at which
the slope of $d_\py(p')$ or $d_\ny(p')$ changes.
Since the slope of $d_\py(p')$ changes at most once within a cell of $\rfmpy$, 
we can find $p'$ in time linear to the complexity of the cells containing $p$ of the maps.
If there are two or more such points, $p$ is 
the point with the maximum $y$-coordinate among them.

There may be no point $p$ satisfying $d_\py(p) = d_\ny(p)$ if 
there is a point $q\in \face \cap L$ such that $d_\py(q')>d_\ny(q')$ for
every point $q'$ lying above $q$, and
$d_\py(q')<d_\ny(q')$ for
every point $q'$ lying below $q$.
We maintain the point $q$ in this case.
Note that $q$ follows a horizontal segment during the plane sweep,
and thus we can find the first point $p$ with $d_\py(p) = d_\ny(p)$ using a horizontal half-line from $q$.

During the plane sweep, $p$ or $q$ moves along $\mbisector(\py,\ny)$ rightwards until it meets the right side of $\face$.
We compute the other bisectors in $\face$ similarly.

We compute the \emph{trace} $\trace(\py,\ny)$ of $p$ and $q$ during the sweep.
Observe that every vertical line intersecting $\face$ also intersects
the trace in one point $t$. Moreover, if the line intersects 
$\mbisector(\py,\ny)\cap \face$, $t$ is the topmost point of the intersection.
Since we have $\rfmpx$ and $\rfmnx$,
we can compute the two traces $\trace(\py,\px)$ and $\trace(\py,\nx)$ similarly.

We observe that each bisector and trace in $\face$ has $O(m)$ complexity.
We get the distance functions using 
$\dspy$, $\dsny$, $\dspx$, and $\dsnx$ which consist of $O(n+m)$ line segments
and support $O(\log (n+m))$ query time.
After computing those distance functions, 
the traces can be constructed in time linear to their complexities.
Thus, in total it takes $O(nm + n \log n + m \log m)$ time to construct the traces for all faces.

\subsection{Partitioning \texorpdfstring{$\face$}{f} into zones}
With the three traces $\trace(\py,\ny)$, $\trace(\py,\px)$, $\trace(\py,\nx)$ in $\face$,
we compute the zone $Z_\py$ in $\face$ corresponding to $\rfmpy$ in $\face$.
Let $T$ be an upper envelope of $\trace(\py,\ny)$, $\trace(\py,\px)$ and $\trace(\py,\nx)$.
Then $Z_\py$ is the set of points lying above $T$ in $\face$. 
See Figure~\ref{fig:sectionfvd}(b). 
The following lemma can be shown using the lemmas in Appendix~\ref{apx:pro.bisectors}.
\begin{lemma}\label{lem:zone}
For any point $p\in Z_\py$, $f(p)=f_\py(p)$. 
\end{lemma}

Similarly, we define the other three zones $Z_\ny$, $Z_\px$, and $Z_\nx$.
Note that $d_\delta(p)>d_{\delta'}(p)$ for every point $p\in Z_\delta$ for 
distinct $\delta,\delta'\in\{\py,\ny,\px,\nx\}$. 
By Lemma~\ref{lem:zone}, $\fvd\cap Z_{\py}$ coincides with $\rfmpy$. 
We copy the corresponding farthest-point map of $\delta$ into $Z_\delta$ 
for each $\delta\in\{\py,\ny,\px,\nx\}$.

We call $\face\setminus(Z_\py\cup Z_\ny\cup Z_\px\cup Z_\nx)$
the bisector zone. Every point $p$ in the bisector zone 
lies on a bisector of two or more maps. Thus, for each bisector of two maps,
we copy one of the maps into the corresponding zone.

\subsection{Gluing along boundaries}
We first glue the zones along their boundaries in each face of $\freespace_V$.
For each edge $e$ incident to two zones, we check whether the two cells
incident to the edge have the same farthest site or not.
If they have the same farthest site, $e$ is not a Voronoi edge of $\fvd$.
Then we remove the edge and merge the cells into one. 
If they have different farthest sites, $e$ is a Voronoi edge of $\fvd$.
This takes $O(nm)$ time in total, which is linear to the number of 
Voronoi edges and cells in $\fvd$.

After gluing zones in every face,
we glue the faces of $\freespace_V$ along their boundaries. 
Since $e$ is a vertical line segment and incident to more than two cells, 
we divide $e$ into pieces
such that any point in the same piece $e'$ is incident to 
the same set of two cells.
If both cells incident to $e'$ have the same farthest site, 
$e'$ is not a Voronoi edge of $\fvd$.
Then we remove the edge and merge the cells.
If they have different farthest sites, $e'$ is a Voronoi edge of $\fvd$.
There are $O(n)$ vertical line segments in $V$
and each of them intersects $O(m)$ cells of $\fvd$, 
so it takes $O(nm)$ time in total.
Then we obtain the geodesic $L_1$ farthest-point Voronoi diagram $\fvd$ 
explicitly. 
See Figure~\ref{fig:sectionfvd}(c).

\begin{theorem}
We can compute the $L_1$ farthest-point Voronoi diagram of $m$ point sites 
in the presence of $n$ axis-aligned rectangular obstacles in the plane
in $O(nm + n \log n + m \log m)$ time and $O(nm)$ space.
\end{theorem}

\begin{corollary}
We can compute the $L_1$ geodesic center of $m$ point sites 
in the presence of $n$ axis-aligned rectangular obstacles in the plane
in $O(nm + n \log n + m \log m)$ time and $O(nm)$ space.
\end{corollary}

\section{Concluding Remarks}
We present an optimal algorithm for computing the farthest-point Voronoi diagram
of point sites in the presence of rectangular obstacles. However, our algorithm may
not work for more general obstacles as it is, because some properties we use for
the axis-aligned rectangles including their convexity may not hold any longer.
Our results, however, may serve as a stepping stone to closing the gap to the optimal
bounds.

\bibliography{paper}

\begin{thebibliography}{10}

\bibitem{aggarwal1989}
A.~Aggarwal, L.J. Guibas, J.~Saxe, and P.W. Shor.
\newblock A linear-time algorithm for computing the {V}oronoi diagram of a
  convex polygon.
\newblock {\em Discrete \& Computational Geometry}, 4(6):591--604, 1989.

\bibitem{alt2005}
H.~Alt, O.~Cheong, and A.~Vigneron.
\newblock The {V}oronoi diagram of curved objects.
\newblock {\em Discrete \& Computational Geometry}, 34(3):439--453, 2005.

\bibitem{aronov1989}
B.~Aronov.
\newblock On the geodesic {V}oronoi diagram of point sites in a simple polygon.
\newblock {\em Algorithmica}, 4(1):109--140, 1989.

\bibitem{aronov1993}
B.~Aronov, S.~Fortune, and G.~Wilfong.
\newblock The furthest-site geodesic {V}oronoi diagram.
\newblock {\em Discrete \& Computational Geometry}, 9(3):217--255, 1993.

\bibitem{bae2009}
S.W. Bae and K.-Y. Chwa.
\newblock The geodesic farthest-site {V}oronoi diagram in a polygonal domain
  with holes.
\newblock In {\em Proceedings of the 25th Annual Symposium on Computational
  Geometry (SoCG)}, pages 198--207, 2009.

\bibitem{moshe2005}
B.~Ben-Moshe, B.K. Bhattacharya, and Q.~Shi.
\newblock Farthest neighbor {Voronoi} diagram in the presence of rectangular
  obstacles.
\newblock In {\em Proceedings of the 13th Canadian Conference on Computational
  Geometry (CCCG)}, pages 243--246, 2005.

\bibitem{moshe2001}
B.~Ben-Moshe, M.J. Katz, and J.S.B. Mitchell.
\newblock Farthest neighbors and center points in the presence of rectangular
  obstacles.
\newblock In {\em Proceedings of the 17th Annual Symposium on Computational
  Geometry (SoCG)}, pages 164--171, 2001.

\bibitem{cheong2011}
O.~Cheong, H.~Everett, M.~Glisse, J.~Gudmundsson, S.~Hornus, S.~Lazard, M.~Lee,
  and H.-S. Na.
\newblock Farthest-polygon {V}oronoi diagrams.
\newblock {\em Computational Geometry}, 44(4):234--247, 2011.

\bibitem{chew1985}
L.P. Chew and R.L. Dyrsdale~III.
\newblock Voronoi diagrams based on convex distance functions.
\newblock In {\em Proceedings of the 1st annual symposium on Computational
  geometry (SoCG)}, pages 235--244, 1985.

\bibitem{choi1998}
J.~Choi, C.-S. Shin, and S.K. Kim.
\newblock Computing weighted rectilinear median and center set in the presence
  of obstacles.
\newblock In {\em International Symposium on Algorithms and Computation}, pages
  30--40. Springer, 1998.

\bibitem{choi1996}
J.~Choi and C.~Yap.
\newblock Monotonicity of rectilinear geodesics in
  \texorpdfstring{$d$}{d}-space.
\newblock In {\em Proceedings of the 12th Annual Symposium on Computational
  Geometry (SoCG)}, pages 339--348, 1996.

\bibitem{rezende1985}
P.J. De~Rezende, D.-T. Lee, and Y.-F. Wu.
\newblock Rectilinear shortest paths with rectangular barriers.
\newblock In {\em Proceedings of the 1st Annual Symposium on Computational
  Geometry (SoCG)}, pages 204--213, 1985.

\bibitem{edelsbrunner1986}
H.~Edelsbrunner and R.~Seidel.
\newblock Voronoi diagrams and arrangements.
\newblock {\em Discrete \& Computational Geometry}, 1(1):25--44, 1986.

\bibitem{fortune1987}
S.~Fortune.
\newblock A sweepline algorithm for {V}oronoi diagrams.
\newblock {\em Algorithmica}, 2(1):153--174, 1987.

\bibitem{giora2009}
Y.~Giora and H.~Kaplan.
\newblock Optimal dynamic vertical ray shooting in rectilinear planar
  subdivisions.
\newblock {\em ACM Transactions on Algorithms}, 5(3):28:1--51, 2009.

\bibitem{hershberger1999}
J.~Hershberger and S.~Suri.
\newblock An optimal algorithm for {Euclidean} shortest paths in the plane.
\newblock {\em SIAM Journal on Computing}, 28(6):2215--2256, 1999.

\bibitem{klein1988}
R.~Klein.
\newblock Abstract {V}oronoi diagrams and their applications.
\newblock In {\em Proceedings of the 4th International Workshop on
  Computational Geometry (EuroCG)}, pages 148--157. Springer, 1988.

\bibitem{lee1980}
D.-T. Lee.
\newblock Two-dimensional {V}oronoi diagrams in the
  \texorpdfstring{$L_p$}{Lp}-metric.
\newblock {\em Journal of the ACM}, 27(4):604--618, 1980.

\bibitem{mitchell1992}
J.S.B. Mitchell.
\newblock \texorpdfstring{$L_1$}{L1} shortest paths among polygonal obstacles
  in the plane.
\newblock {\em Algorithmica}, 8(1--6):55--88, 1992.

\bibitem{oh2019}
E.~Oh.
\newblock Optimal algorithm for geodesic nearest-point {V}oronoi diagrams in
  simple polygons.
\newblock In {\em Proceedings of the 30th Annual ACM-SIAM Symposium on Discrete
  Algorithms (SODA)}, pages 391--409, 2019.

\bibitem{oa20geodesic}
E.~Oh and H.-K. Ahn.
\newblock Voronoi diagrams for a moderate-sized point-set in a simple polygon.
\newblock {\em Discrete \& Computational Geometry}, 63(2):418--454, 2020.

\bibitem{ola20fgeodesic}
E.~Oh, L.~Barba, and H.-K. Ahn.
\newblock The geodesic farthest-point {V}oronoi diagram in a simple polygon.
\newblock {\em Algorithmica}, 82(5):1434--1473, 2020.

\bibitem{papadopoulou2013}
E.~Papadopoulou and S.K. Dey.
\newblock On the farthest line-segment {V}oronoi diagram.
\newblock {\em International Journal of Computational Geometry \&
  Applications}, 23(06):443--459, 2013.

\bibitem{papadopoulou2001}
E.~Papadopoulou and D.T. Lee.
\newblock The \texorpdfstring{$L_\infty$}{L0} {V}oronoi diagram of segments and
  {VLSI} applications.
\newblock {\em International Journal of Computational Geometry \&
  Applications}, 11(05):503--528, 2001.

\bibitem{shamos1975}
M.I. Shamos and D.~Hoey.
\newblock Closest-point problems.
\newblock In {\em Proceedings of the 16th IEEE Annual Symposium on Foundations
  of Computer Science (FOCS)}, pages 151--162, 1975.

\bibitem{wang2021}
H.~Wang.
\newblock An optimal deterministic algorithm for geodesic farthest-point
  {V}oronoi diagrams in simple polygons.
\newblock In {\em Proceedings of the 37th International Symposium on
  Computational Geometry (SoCG)}, pages 59:1--59:15, 2021.

\end{thebibliography}

\end{document}